\documentclass[12pt,a4paper]{amsart}

\usepackage{bm}

\usepackage{amsfonts}
\usepackage{amsmath}
\usepackage{amsthm}
\usepackage{amscd}
\usepackage{enumerate}
\usepackage{amssymb}
\usepackage{graphicx, wrapfig}
\usepackage[font=scriptsize]{caption}
\usepackage{natbib}
\usepackage{subcaption}
\usepackage{hyperref}
\usepackage{verbatim}

\newtheoremstyle{mytheorem}
{3pt}
{1pt}
{\itshape}
{-0pt}
{\bf}
{}
{1em}
{}
\newtheoremstyle{mydefinition}
{3pt}
{6pt}
{\itshape}
{-0pt}
{\bf}
{}
{1em}
{}
\newtheoremstyle{myremark}
{3pt}
{1pt}
{\rm}
{-0pt}
{\bf}
{}
{1em}
{}
\theoremstyle{mytheorem}
\newtheorem{Theorem}{Theorem}[section]
\newtheorem{Proposition}[Theorem]{Proposition}

\theoremstyle{mydefinition}
\newtheorem{Definition}[Theorem]{Definition}

\newtheorem{Remark}[Theorem]{Remark}

\newtheorem{Algorithm}[Theorem]{Algorithm}

\bibpunct{(}{)}{; }{a}{,}{,}

\def\R{\mathbb R}
\def\N{\mathbb N}

\def\K{\mathbb K}
\def\to{\rightarrow}

\selectfont

\usepackage[top=2cm, bottom=2cm, left=2cm, right=2cm,]{geometry}

\begin{document}

\title{Mobility decisions, economic dynamics and epidemic}

 \author[G. Fabbri]{Giorgio Fabbri$^a$}
 \thanks{$^{b}$Univ. Grenoble Alpes, CNRS, INRAE, Grenoble INP, GAEL, 38000 Grenoble, France.}
 \address{G. Fabbri, Univ. Grenoble Alpes, CNRS, INRA, Grenoble INP, GAEL - CS 40700 - 38058 Grenoble CEDEX 9, France. The work of Giorgio Fabbri is partially supported by the French National Research Agency in the framework of the ``Investissements d'Avenir'' program (ANR-15-IDEX-02) and of the center of excellence LABEX MME-DII (ANR-11-LABX-0023-01).}
 \email{giorgio.fabbri@univ-grenoble-alpes.fr}.

 \author[S. Federico]{Salvatore Federico$^b$}
 \thanks{$^{b}$Universit\'a degli Studi di Genova, Dipartimento di Economia.}
 \address{S. Federico, Universit\`a degli Studi di Genova, Dipartimento di Economia. Via Vivaldi, 5, Darsena, 16126, Italy.}
 \email{salvatore.federico@unige.it}

 \author[D. Fiaschi]{Davide Fiaschi$^c$}
 \thanks{$^{c}$Universit\'a degli Studi di Pisa, Dipartimento di Economia e Management.}
 \address{D. Fiaschi, Universit\`a degli Studi di Pisa, Dipartimento di Economia e Management. Via Ridolfi 10, 56124 Pisa (PI), Italy.}
 \email{davide.fiaschi@unipi.it }

 \author[F. Gozzi]{Fausto Gozzi$^d$}
 \thanks{$^{d}$Dipartimento di Economia e Finanza, LUISS \emph{Guido Carli}, Roma.}
 \address{F. Gozzi, Dipartimento di Economia e Finanza, Libera Universit\'a degli Studi Sociali Guido Carli, Roma.} \email{fgozzi@luiss.it}

\maketitle

\begin{abstract}
{We propose a model, which nests a susceptible-infected-recovered-deceased (SIRD) epidemic model into a dynamic macroeconomic equilibrium framework with agents' mobility. The latter affect both their income and their probability of infecting and being infected. Strategic complementarities among individual mobility choices drive the evolution of aggregate economic activity, while infection externalities caused by individual mobility affect disease diffusion. The continuum of rational forward-looking agents coordinates on the Nash equilibrium of a discrete time, finite-state, infinite-horizon Mean Field Game.}

\smallskip



\noindent {We prove the existence of an equilibrium and provide a recursive construction method for the search of an equilibrium(a), which also guides our numerical investigations.}

\smallskip

\noindent  {We calibrate the model by using Italian experience on COVID-19 epidemic and we discuss policy implications.}

\bigskip

\noindent \textit{Keywords}: {mean field game, strategic complementarities, ESIRD, COVID-19.} 

\noindent \textit{JEL Classification}: E1, H0, I1, C72, C73, C62.
\end{abstract}

\newpage

\section{Introduction}

We propose an integrated assessment model, denoted by ESIRD, encompassing a susceptible-infected-recovered-deceased (SIRD) epidemic model and a dynamic macroeconomic equilibrium economic model, where \textit{mobility choices} of forward-looking agents affect both income (and consumption) and the spread of epidemic. A calibrated version of the model illustrates the possibilities to use the model to design an efficient policy of state-of-epidemic-dependent mobility restrictions.

\medskip

Pandemic crisis has shown that sudden drops in individual mobility have a substantial negative consequence on aggregated income and consumption \citep{organizaccao2020evaluating}. The decrease of individual mobility along the COVID-19 crisis has been the joint outcome of individual decisions, caused by the diffusion of infection, and of containment measures imposed by national authorities (lockdown, curfew, etc.). In turn, a reduction in individual mobility brings down individual income \citep{huang2020quantifying} as well as epidemic dynamics, being higher individual mobility associated to a higher probability of infecting and being infected \citep{nouvellet2020report}. Therefore, entangled externalities and ``equilibrium'' effects are at work; more precisely, individual mobility decisions display i) \textit{strategic complementarities} with mobility choice of other agents, because the marginal impact on individual income of individual mobility is increasing in the  mobility \citep{bulow1985multimarket,cooper1988coordinating}; and, ii)\textit{ negative externalities} on contagion dynamics, because of agents in their mobility choices internalize the risk of being infected, but not the effect of infecting other people \citep{bethune2020covid}.\footnote{Another possible source of externality, the healthcare congestion, is analyzed by \cite{jones2020optimal}.} 

\medskip

In the model we focus on \textit{short-term} mobility. 
Epidemic dynamics is driven by a generalized version of the SIRD model, where the average number of contacts per person per time is endogenous, as well as the transition rate (i.e., the flow of new infected), and depends on the mobility choices of agents.

Agents maximize an inter-temporal discrete time utility function considering consumption and mobility costs. Their choice of mobility for work (respectively for consumption) depends on their state (susceptible, infectious, or recovered), the aggregate level of economic activity, the current and future policies on mobility restrictions, and on their future utility, which, in turn, depends on the probabilities of being infected in the future and on the future economic dynamics. At each time, aggregate economic activity (consumption) depends on the state of the epidemic and on the individual mobility choices.

We set the agent's problem as a game with a continuum of players in a finite state space (the four states of agents) and, in particular, the model can be seen as a discrete time, finite state, infinite horizon Mean Field Game (MFG) \citep{LasryLions07}. The notion of equilibrium used in the paper is basically borrowed -- even if re-elaborated -- from \citealp{jovanovic1988anonymous} (Definition \ref{def:eq}), which we show to be equivalent to the more common notion of Nash equilibrium of our Mean Field Game (Proposition \ref{prop:equivalence}). We then provide the proof of the existence of the equilibrium for our Mean Field Game (Theorem \ref{th:existence}), and finally propose a recursive algorithm to identify and then numerically simulate such equilibrium (Section \ref{sec:alg} and Theorem \ref{thm:verif}).

\medskip

MFG literature deals with the behavior of Nash equilibria in differential games as the number of agents becomes large.
There is extensive recent research activity on MFGs starting from the pioneering works of Huang, Malhamé and Caines \citep{HuangMalhameCaines06} and, independently, at the same time by Lasry and Lions \citep{LasryLions06I, LasryLions06II, LasryLions07}.
In the large population limit, one expects to obtain a game with a continuum of agents where, like in our case, the effects on the decision of any agent from the actions of the other agents are experienced through the statistical distribution of states.
Since perturbations from the strategy of an agent does not influence the statistical states' distribution, the latter acts as a parameter in each agent's control problem.

\medskip

We calibrate the model by using Italian experience on COVID-19 epidemic in the period February 2020-May 2021. Numerical explorations under different configurations of state-of-epidemic-dependent mobility restrictions highlight the presence of a trade-off between economic losses and fatalities due to the pandemic, that is, of a pandemic possibilities frontier as in \cite{kaplan2020great} and \cite{acemoglu2020multi}. However, we argue that policy evaluation should take into account two additional directions. The first is related to the share of susceptible at the end of the period of evaluation, which can favor a new outbreak of epidemic in the future without an efficient vaccine. The second is the social feasibility of prolonged mobility restrictions \citep{vollmer2020report}.

\medskip

Our paper makes four main contributions to literature.
The first is to the epidemiological-macroeconomic literature, which has recently boomed following the COVID-19 outbreak. Its main goal is to produce integrated assessment models, where the economic dynamics complements epidemiological models. In particular, a strand of literature focuses on optimal policy problem from a planner's perspective without modeling individual behavior (see, e.g., \citealp{alvarez2020simple,piguillem2020optimal, moser2020pandemic, atkeson2020will}), while another one considers forward-looking agents and market determination of good and factor prices, as in \cite{eichenbaum2020macroeconomics}, \cite{toxvaerd2020equilibrium}, \cite{jones2020optimal} and \cite{kaplan2020great}.
With respect to these contributions, we explicitly consider agents' (short-term) mobility. There are several good reasons for this focus: (i) in the epidemiological literature, mobility is (not surprisingly) identified  as the key variable in containing an epidemic \citep{nouvellet2020report}; (ii) mobility is an easily measurable variable and many datasets are freely available; 
and, (iii) since mobility was/is the primary focus of several restrictive policies imposed by governments, the proposed framework is a natural candidate to evaluate past and future policies on mobility restrictions.
As already argued, focusing on mobility implies taking into account non-market interactions among individual choices: the presence of strategic complementarities in individual decisions is another element of novelty in our epidemiological-macroeconomic model. This introduces substantial difficulties in the mathematical study of the model, which arise e.g. from proving the existence of a Nash equilibrium.

An advantage of our analysis is to consider individuals with a long (infinite) time horizon. This is crucial for understanding the interaction between the change in death risk (whose effects should be evaluated over years), and the epidemic dynamics (whose effects should be measured over days). For example, in a two (or three)-period model (as for instance \citealp{bandyopadhyay2021learning} or \citealp{bhattacharya2021rational}), a strategy to reduce mobility (and consumption) in the short run so that to decrease the death risk and wait for the end of the epidemics cannot be correctly evaluated. Similar situation applies for the non-linear dynamics of the epidemics and bringing the model to empirical data, which would also be problematic.

\medskip

The second contribution is on methodology.
We have discussed above that our model belongs to the class of discounted infinite horizon, discrete time, finite state space MFG. To the best of our knowledge, this does not fall into the classes already studied in the literature, among which \cite{GomesMohrSousa12}, \cite{DoncelEtAl19}, \cite{HadikanlooSilva19}, and \cite{BonnansEtAl21,Wiecek19}. \cite{HadikanlooSilva19} and \cite{BonnansEtAl21} consider only finite horizon problems, while \cite{GomesMohrSousa12} (and similarly \citealp{Wiecek19}) consider infinite horizon problems of ergodic type or with entropy penalization, where the dependence of the agents' utility from the choices of the other agents is more regular than in our model. 
\cite{DoncelEtAl19} consider an infinite horizon MFG, but where agents' cost does not depend on the strategies of the other agents, which instead happens in our model for the presence of strategic complementarities. Hence, our theorems of existence of an equilibrium and the recursive construction of an equilibrium are to be considered a novelty.

\medskip

We also contribute to the literature focusing on the endogenous determination of the infection rate and the reproduction rate of an epidemic \citep{avery2020economist}. Infection rate depends on many aggregate factors (climate, geography, health system, etc.), but also crucially revolves around individual choices. Several approaches have been proposed to endogenize infection rates, among which a purely epidemiological approach as \cite{fenichel2013economic}, and a behavioral approach as, for example, in \cite{engle2020behavioral} and \cite{bisin2021spatial}. \cite{farboodi2020internal}, \cite{toxvaerd2020equilibrium}, and \cite{eichenbaum2020macroeconomics} are instead more in line with our approach, developing settings where forward-looking individuals chose their actions facing an epidemic-economic trade-off. However, no paper directly models mobility choices of individuals while considering strategic complementarities and negative externalities in an infinite horizon equilibrium setting as a way to explain the dynamics of infection rate during the pandemic. The advantage of our approach is evident in the interpretation of results, allowing to directly correlate mobility and infection rates, and in the possibility to bring the model to data.

\medskip

The final contribution is to the literature on the effect of epidemics diffusion on mobility (see, e.g., \citealp{goolsbee2021fear} and \citealp{meloni2011modeling} and \citealp{nouvellet2020report} for an epidemiological perspective). 
Our contribution provides a theoretical framework to evaluate restrictive policies going beyond the simple trade-off economic losses/fatalities as prospected in \cite{kaplan2020great}, \cite{acemoglu2020multi}, and \cite{gollier2020cost}. It makes it possible, for instance, to take into account other key dimensions regarding the social feasibility of policies in the evaluation, such as the fragility of post-lockdown situations with a high risk of new outbreaks, and the sustainability of health systems (see, in particular, Sections \ref{sec:calibrationValidationModel} and \ref{sec:questionModel}). 


\medskip

The paper is organized as follows: Section \ref{sec:model} presents the model, Section \ref{sec:singleagent} focuses on the agent's optimization problem while Section \ref{sec:equilibrium} provides a recursive construction method for the search of an equilibrium(a).
Section \ref{sec:calibrationValidationModel} calibrates the model to Italian data;  Section \ref{sec:questionModel} uses the model to investigate the effects of policies aiming at mitigating epidemics and their effects on economic activity; Section \ref{sec:concludingRemarks} concludes.

\section{The epidemiologic-economic dynamic model}
\label{sec:model}

We consider an infinite horizon discrete time ($t=0,1,2...$) world with a continuum set of agents, whose individual mass is equal to zero so that the actions of a single agent do not modify the evolution of the global epidemic state and of the aggregate economic variables.

As in the classical SIRD framework (\citealp{chowell2016mathematical}), at each time, the \textit{health status} $k$ of an agent can be: susceptible ($k=S$); infected ($k=I$); recovered ($k=R$); and deceased $(k=D)$.
We then denote the set of possible health status by $\mathbb{K}$, i.e.
\[
\mathbb{K}:=\left\{S,I,R,D\right\}.
\]
We denote by $\mu(t,k)$ the share of the population in the health status $k$ at time $t$ and by $\bm{\mu}(t)$ the four-dimensional vector $\bm{\mu}(t) = (\mu(t,S),\mu(t,I),\mu(t,R),\mu(t,D))$ representing the \textit{health status distribution of the population}.\footnote{The sum of the components of $\bm{\mu}(t)$ is always equal to $1$; hence, $\bm{\mu}(t)$ is can be seen as a probability measure on $\mathbb{K}$. {Roughly speaking, this means that our state of the world is described up to sets of agents of measure zero.}}

\subsection{The agents' utility}
\label{sub:utility}
Each agent chooses at each time $t$ their mobility rates (whose maximal value is w.l.o.g. normalized to $1$) for production, $\vartheta_p(t)$, and for consumption, $\vartheta_c(t)$.

The instantaneous utility at time $t$ of the agent in the health state $k(t)\in\K$, undertaking the actions $\bm{\vartheta}(t):=(\vartheta_p(t),\vartheta_c(t))\in[0,1]^{2}$ is equal to $0$ if $k(t)=D$, otherwise,
\[
u(t,c(t),k(t),\bm{\vartheta}(t)) :=
\ln c(t)
- \gamma_p\left(t,k(t),\bm{\mu}(t)\right) \vartheta_{p}(t) - \gamma_c\left(t,k(t),\bm{\mu}(t)\right) \vartheta_{c}(t)- M.
\]
In the above expression, $c(t)$ is the individual consumption, $M\in\R$ is the (exogenous) constant utility of state deceased, which ``normalizes the utility of nonsurvival to zero'' \citep[p. 2]{rosen1988value}, and  $\gamma_p\left(t,k(t),\bm{\mu}(t)\right)$ and $\gamma_c\left(t,k(t),\bm{\mu}(t)\right)$ are, respectively, the marginal utility cost to move into the labour market (and in general for the movements related to the productive activities of the agent) and to move into the consumption market (or, in general, for the movements related to the individual consumption).

{The cost of moving must be understood in a broad sense, taking into account several elements: the cost in the narrow sense,  e.g., for gasoline, but also the difficulties in the movement related to the status "infected", i.e., direct virus-related impediments; moreover, the stress generated by the likelihood of coming into contact with the virus, i.e., the ``psychological'' costs of bearing the risk to become infected; and, finally, the penalties related to legal constraints, i.e., government-imposed restrictions to movements.}

{In particular,} the functions $\gamma_p$ and $\gamma_c$ will be used to model public policies for mobility restriction. For this reason, they may depend explicitly on time $t$ (in the case of policies that intervene at exogenously fixed times) or on the state of the epidemic (for example, in the case of policies that change endogenously depending on the severity of the epidemiological situation). The mobility cost structure is known by agents who will incorporate, in their inter-temporal choices, future policy changes (both exogenous and endogenous). 

We make the following assumptions on the marginal utility cost of mobility:
\[
\gamma_{p}(t,R,\bm\mu)\leq \gamma_{p}(t,S,\bm\mu)\leq \gamma_{p}(t,I,\bm\mu) \text{ and }  \gamma_{c}(t,R,\bm\mu)\leq \gamma_{c}(t,S,\bm\mu)\leq \gamma_{c}(t,I,\bm\mu),
\]
for any $\bm\mu$ and $t$.
{This ranking naturally follows from the broad view underlining the mobility costs in our model. While we expect that the costs in the narrow sense associated with mobility do not depend on the agent's status, the ``psychological'' costs associated to the possibility of coming into contact with the virus are greater in susceptibles than in recovered agents (first inequality). On the other hand, both the possibly massive physical impediments associated with an ongoing infection, and the greater legal penalties (quarantine), suggest that the marginal cost for the infected agents is greater than for the other agents (second inequality).}

\medskip

As described in detail in Subsections \ref{sub:mobilityandepidemics}, at each time any susceptible individual has a certain probability of becoming infected, and each infected individual has a certain probability of dying and of recovering; hence, the evolution of the individual health status $k(t)$ is represented by a discrete stochastic process. The goal of each agent will be to maximize their total expected inter-temporal utility given by:
\begin{equation}
\label{eq:intertemporalUtility-pre}
\mathbb{E} \left [ \sum_{t=0}^{\infty} {(1-\rho)}^{t} u(t,c(t),k(t),\bm{\vartheta}(t)) \right],
\end{equation}
where $(1-\rho)\in(0,1)$ is the exogenous discount factor.

\subsection{Consumption and mobility\label{sub:consumption}}

We suppose that the opportunity to move into the consumption market produces a benefit for agents. Moving can indeed allow access to a greater number of goods and services and to a wider variety, satisfying more precisely the needs of the agent or finding equivalent goods with inferior prices. Alternatively, we can suppose that the effective consumption is affected by the mobility/time dedicated to the consumption activity \citep{steedman2001consumption}.
{To formalize this idea as simply as possible, we suppose that the ``generalized cost'' faced by the agent for one unit of consumption good depends on its (consumption-related) mobility choice $\vartheta_{c}(t)$, in particular:
\[
G(\vartheta_{c}(t)) = \frac{1}{P_{0} + P_{1} \vartheta_c{(t)}},
\]
where
\[
{P_{0}, P_{1}\geq 0}
\]
are exogenous constants. $P_{0}$ reflects the presence of some baseline/no-movement cost, while $P_{1}$ measures the the decline in the cost due to movement. As will become clear below, the specific form chosen for $G$ responds to a particularly simple and transparent specification of the model.}

{In the model, we do not consider the saving (and therefore the dynamics of accumulation of capital) and therefore we impose, at every time, that the individual income is entirely destined to the expenditure for consumption, i.e.
\[
 y(t) = G(\vartheta_{c}{(t)}) c(t),
\]
where we denoted by $y(t)$ the individual income at time $t$.}
This implies that:
\begin{equation}
\label{eq:consumptionRate}
c(t) =  \left(P_0 + P_1 \vartheta_{c}(t) \right) y(t),
\end{equation}
i.e. the consumption is decided by the level of income, but also by the mobility for consumption.

\subsection{Income and mobility}
\label{sub:incomemobility}

As for the need of mobility for consumption, we assume that mobility affects agent's income, in particular a greater mobility positively contributes to income. The idea here is intuitive: some jobs/activities require the presence of the worker and can, or cannot, only be partially carried out by remote work. We also suppose that income is affected by health conditions of agents (obviously, sick people are less productive than healthy ones), and that productivity, and therefore, agent's income also depends on the macroeconomic conditions, so that a greater macroeconomic activity, \textit{ceteris paribus}, will lead to higher agent income. In the model, $Z(t)$ will denote the level of macroeconomic activity at time $t$ and its dependence on agents' choices will be discussed shortly. All in all, we suppose that the individual income has the following form:
\begin{equation}
y(t) =Z(t)\left(A_{0}^k + A_{1}^k \vartheta_p(t)\right),
\label{eq:agentIncome}
\end{equation}
{where $A_0^k$ and $A_1^k$ are positive exogenous constants, depending on the health status $k$ of the agent. This specification represents in the simplest (bilinear) way the complementarity between the individual contribution $\left(A_{0}^k + A_{1}^k \vartheta_p(t)\right)$ and the aggregate economic activity level $Z(t)$. In Eq. (\ref{eq:agentIncome}) the individual contribution has a linear specification that captures the mechanism described without increasing excessively the formal complexity of the model.}

We will suppose that $A_1^S= A_1^R$ so we will denote this value by $A_1^{SR}$, and we will suppose that
\[
0< A_0^{I}\leq A_0^{SR} \quad \text{ and } \quad 0\leq A_1^I\leq A_{1}^{SR},
\]
where the second inequalities reflect the fact that healthy (susceptible or recovered) agents are more productive than infected.

From (\ref{eq:consumptionRate}), the consumption of the agent in the health state $k$, when the epidemic is in the state $\bm\mu(t)$ and they undertake the production-consumption choices $\bm\vartheta(t)=(\vartheta_{c}(t),\vartheta_{c}(t))$, is then given by
\begin{equation}
\label{eq:consumtption}
c(t) = Z(t) \left( A_{0}^k + A_{1}^k \vartheta_p(t)\right) \left(P_{0} + P_{1} \vartheta_c(t) \right).
\end{equation}

The level of macroeconomic economic activity $Z(t)$ depends on the choices of all agents on their mobility for the participation in the productive activities, and thus it presents strategic complementarities. More precisely, we will suppose that it has the following shape:
\begin{eqnarray}\label{zzz}
Z(t)  := \phi\bigg( \mu(t,S)\bar\vartheta_p(t,S), \,\,\mu(t,I)\bar\vartheta_p(t,I), \,\, \mu(t,R)\bar\vartheta_p(t,R)\bigg),
\end{eqnarray}
where $\phi:[0,1]^{3}\to (0,\infty)$ is non-decreasing in all the components and such that $\phi(0,0,0)>0$ and $\bar\vartheta_p(t,S)$ (respectively $\bar\vartheta_p(t,I)$ and $\bar\vartheta_p(t,R)$) is the average productive-mobility choice of susceptible (respectively infected, recovered) agents. In the following, we will focus on symmetric equilibria where all individuals of the same health status behave in the same way; hence, along the equilibrium, $\bar\vartheta_p(t,S)$, $\bar\vartheta_p(t,I)$ and $\bar\vartheta_p(t,R)$ will also be the (optimal) choice of any single agent.

\subsection{Agents mobility and epidemic dynamics}
\label{sub:mobilityandepidemics}
We model the evolution of the size of health classes, that is, the shares of population with different health status, following a standard SIRD model \emph{without vital dynamics} (newborns are not considered and people die only because of the virus) adjusted for the mobility choices of the agents.

To make the point clearer, we recall that in the standard SIRD model the number of new infected agents is given by
\begin{equation}
\beta\frac{I(t) S(t)}{N(t)},
\label{eq:betaSIRD}
\end{equation}
where $I(t)$ (respectively $S(t)$, $N(t)$) is the number of infected agents (respectively susceptible agents, total number of agents) at time $t$ and $\beta$ is an exogenous factor representing the average number of contacts 
per agent per time.

In the standard SIRD model, $\beta$ is constant and is exogenous with respect to the state of epidemics and agents' choices. The idea behind this formulation is that people meet by chance independently of their epidemiological status; hence, the probability of a susceptible agent meeting an infected agent and getting infected at time $t$ is
\[
\beta\frac{I(t)}{N(t)} = \beta \mu(t,I).
\]
As a result, in the standard SIRD model, the share of the \textit{new} infected individuals at time $t$ is
\[
\beta\frac{I(t)}{N(t)} \frac{S(t)}{N(t)} = \beta \mu(t,I)\mu(t,S).
\]
Based on the idea that the number of contacts depends on the mobility of agents, we enrich this formulation adjusting the parameter $\beta$ for the agents' mobility choices.
In particular, we observe that it is natural to suppose that the number of contacts is proportional to the distance covered by agents; for example, an agent walking 200 meters in a street would meet twice as many other agents than if they walked 100 meters. For the same reason, the number of contacts is proportional to the average distance covered by other agents given the mobility of the agent.

Therefore, since the maximal mobility is normalized to one and distinguishing the mobility for production and for consumption, the probability of a susceptible agent with mobility $(\vartheta_p(t),\vartheta_p(t))$ meeting an infected agent and getting infected is modeled as
\begin{equation}
\label{eq:taudef}
\tau(t)=\left (\beta_{p}\bar\vartheta_{p}(t,I)\vartheta_p(t) +\beta_{c}\bar\vartheta_{c}(t,I)\vartheta_p(t) \right )\mu(t,I),
\end{equation}
where $\beta_p, \beta_c>0$ are given constants that we assume to satisfy the condition $\beta_p+\beta_c<1$.

Taking the average over the population of susceptibles, and multiplying by the portion of susceptibles among the population, we find the share of the new infected agents; the latter represent the (negative) variation in the share of the susceptible population, that is\footnote{The assumptions of zero mortality for reasons different from the virus and of the zero natality are implicit in \eqref{eq:dynamicsMassSusceptibles}.}
\begin{equation}
\mu(t+1,S) = \mu(t,S) - \beta(t)\mu(t,S)\mu(t,I),
\label{eq:dynamicsMassSusceptibles}
\end{equation}
where
\begin{equation}
\beta(t) :=\beta_{p}\bar\vartheta_{p}(t,I){\bar\vartheta}_p(t,S) +\beta_{c}\bar\vartheta_{c}(t,I){\bar\vartheta}_c(t,S).
\label{eq:infectionRateMobility}
\end{equation}
Therefore, in our ESIRD (economic SIRD) model, $\beta(t)$ of \eqref{eq:infectionRateMobility} is the counterpart of $\beta$ in the SIRD model of \eqref{eq:betaSIRD}.

Apart for the role of mobility in $\beta(t)$, we will stick to the classic structure of the SIRD model, and we suppose that $\pi_D$ (respectively $\pi_R$) is the probability of an infected agent to die (respectively to recover) at each time. At the aggregate level, this means that a portion $\pi_D$ (respectively $\pi_R$) of infected agents die (respectively recover) at each time. Hence, the evolution of the health status distribution of population in our model is as follows:
\begin{equation}
\label{eq:evolu}
\bm{\mu}(t+1)=\bm{Q}(t)\bm{\mu}(t),
\end{equation}
where
$$
\bm{Q}(t):=\left(
\begin{array}{cccc}
1- \beta(t)\mu(t,I) & 0 & 0 & 0 \\
\beta(t)\mu(t,I)  & 1- \pi_R -\pi_D  & 0 & 0 \\
0 & \pi_R & 1 & 0  \\
0 & \pi_D & 0 & 1 \\
\end{array}
\right).
$$

From \eqref{eq:infectionRateMobility} we observe that the dependence of $\beta(t)$ on agents' mobility is proportional to the product of individual mobilities, which generates \textit{strategic complementarities} in the mobility choices\footnote{{In the model the individual mass is zero, hence individual decisions do not affect the dynamics of the aggregate variables. Here, we refer to \emph{mobility choices} in the sense of the average mobility choices. This fact will be further stressed in the next section.}} with aggregate negative effects. In particular, infected agents do not internalize the effect of their mobility choice on the infection rate of susceptible agents, and both susceptible and infected agents do not internalize the effect of the increased spread of the pandemic on the level of macroeconomic activity $Z(t)$. Therefore, in the \textit{decentralized} equilibrium, the agents’ mobility is too high with respect to the optimal \textit{social} mobility.

\section{The agent's optimization problem \label{sec:singleagent}}

We now look at the optimization problem of a single agent. As previously discussed, the zero-mass agent assumption implies that the individual choices of any specific agent do not modify the macro variables and, in particular, the evolution of the epidemic according to \eqref{eq:evolu}. The latter only depends on the average choices of each group defined by agents' health status. This means that agents take the average strategies $\bar{\bm{\vartheta}}(t)$ and the dynamics of $\bm{\mu}(t)$ as given when they make their decisions, that is, we are considering a \textit{Mean Field Game} \citep{LasryLions07}. At the equilibrium, we will impose that optimal individual decisions coincide with the average decisions of the corresponding group defined by agents' health status.

The epidemic dynamics $\bm{\mu}(t)$ does not depend on the choices of the single agent; however, the evolution of their epidemic status does. In particular, as we have already discussed in Section \ref{sec:model}, the probability of a susceptible agent getting infected is given by the \textit{endogenous} probability $\tau(t)$ defined in \eqref{eq:taudef}, while the probabilities of an infected agent dying and recovering are \textit{exogenous} and equal to $\pi_D$ and $\pi_R$, respectively. Hence, the state of the agent $k(t)$ is represented by a controlled \textit{Markov Chains}, whose \textit{transition kernel} at each time $t$ is given by:
\[
 \bm{q}(t) =
\left(
\begin{array}{cccc}
p_{SS}(t) & p_{IS}(t) & p_{RR}(t) & p_{DS}(t)  \\
p_{SI}(t)  & p_{II}(t)  & p_{RI}(t)  & p_{DI}(t) \\
p_{SR}(t)  & p_{IR}(t)  & p_{RR}(t)  & p_{D}(t)  \\
p_{SD}(t)  & p_{ID}(t) & p_{RD}(t)  & p_{DD}(t)  \\
\end{array}
\right)
=
\left(
\begin{array}{cccc}
1- \tau(t)& 0 & 0 & 0 \\
\tau(t) & 1- \pi_R -\pi_D  & 0 & 0 \\
0 & \pi_R & 1 & 0  \\
0 & \pi_D & 0 & 1 \\
\end{array}
\right),
\]
where $p_{k_1 k_2}(t)$ is the probability to switch from the status $k_1$ at time $t$ to the status $k_2$ at time $t+1$. Even if not emphasized in the notation, $\bm{q}$ depends on the individual decisions $\bm{\vartheta}(t)$ and on the average decisions of other agents $\bar {\bm \vartheta}(t)$.

Since we rely on dynamic programming, we let the initial time and state vary. Hence, we assume that the agent starts at time $t_{0}\in\N$ in the state $k(t_0)\in\K$, where $(t_{0},k(t_0))\in\N\times \K$, and that they choose their strategies in the set:
\[
\mathcal{A}(t_{0}):=
\Big\{\bm{\vartheta}=({\vartheta}_p,{\vartheta}_c): \{t_0,t_{0}+1,...\}\times \mathbb{K} \to [0,1]^2 \ \mbox{s.t.}\ \ {\vartheta}(\cdot,D)=(0,0)\Big\}.
\]
In general, the set of admissible strategies depends on the time $t_0$ and we should denote the set of strategies by $\mathcal{A}(t_0)$. At each time, the strategy $\bm\vartheta$ can be chosen from all pairs $(\vartheta_p, \vartheta_c) \in [0,1]^2$, so with a slight abuse of notation (making abstraction of the translation, for which the strategy at time $t_0$ is defined only for $t\geq t_0$) we will denote $\mathcal{A}$ as the set of admissible strategies. In the set of strategies, each agent includes a complete plan of action for: i) the initial health states different from the actual one of the same agent; and, ii) all possible future health statuses, even though some of these are not attainable; for example, recovered agents cannot become susceptible or infected in the future.


The counterpart of the target (\ref{eq:intertemporalUtility-pre}) starting from $(t_{0},k(t_0))$ depending on the initial health status distribution $\bm{\mu}(t_0)$ and on the average strategies $\bar{\bm{\vartheta}}(t,k) $ specified for all $t\geq t_0$ and $k \in \mathbb{K}$ is
\[
J(t_{0},k(t_0),\bm{\mu}(t_0),\bar{\bm\vartheta}(\cdot,\cdot);\bm{\vartheta}(\cdot,\cdot)) := \mathbb{E} \left [ \sum_{t=t_{0}}^{\infty} {(1-\rho)}^{t-t_{0}} u(t,c(t),k(t), \bm\vartheta(t,k(t))) \right],
\]
where $c(t)$ is just an abbreviation.\footnote{In particular, from (\ref{eq:consumtption}), $c(t)= Z(t) \left( A_{0}^k + A_{1}^k \vartheta_p(t)\right) \left(P_{0} + P_{1} \vartheta_c(t) \right)$ and $Z(t)$ is given by (\ref{zzz}). Hence, $c(t)$ does depend on $\bm\bar{\bm\vartheta}$ and $\bm{\mu}$.}

The value function of the agent is defined as
\begin{equation}
\nonumber
V(t_{0},k(t_0),\bm{\mu}(t_0),\bar{\bm\vartheta}(\cdot,\cdot)):=\sup_{\bm{\vartheta}(\cdot,\cdot)\in\mathcal{A}} J(t_{0},k(t_0),\bm{\mu}(t_0),\bar{\bm\vartheta}(\cdot,\cdot);\bm{\vartheta}(\cdot,\cdot)).
\end{equation}
According to the \textit{dynamic programming principle}, the value function is a solution (possibly not unique) to the \textit{Bellman equation} (with unknown $v$)
\begin{align}\label{DPE}
v(t_{0},k(t_0)) = \sup_{\bm\vartheta\in [0,1]^2}
\sum_{k\in\K} p_{k(t_0)k}(t_{0})\big[ u(t_{0}, c(t_0), k(t_0),\bm\vartheta)
+ (1-\rho)v(t_{0}+1,k) \big].
\end{align}

\section{The limits of our modelling strategies}

In our model formulation, we adopt some shortcuts that need more detailed discussion.
The positive relationship between utility and individual mobility is to be considered a reduced form of the result of solving the equilibrium of an economy populated by firms producing heterogeneous goods and services in different locations and by consumers with heterogeneous preferences incurring moving costs in their search for the best consumption basket. In equilibrium, the reduced mobility should determine higher prices as the result of lower competition among firms; additionally, the same quantity of consumption should also lead to a lower utility for the possible mismatch between the consumers' heterogeneous preference and the specific local supply of goods and services.
A complementary explanation of the positive effect of mobility on individual utility is that reduced mobility constrains the capacity of expenditure of individuals, which turns out as forced saving. In our framework, where saving is absent, the reduced mobility, therefore, corresponds to an increasing gap between income and consumption, that is, between the latter and utility.
Also, the relationship between mobility and individual income is to be taken as a reduced form of the equilibrium of an economy, where the place of residence and place of work differ (i.e., there exists commuting); where the production activity needs some mobility, for example, the need to transport commodities among different plants; and where the place of production and the place of sale differ, which is the most common case. As a result, in equilibrium, reduced mobility leads to a decrease in economic activity. 
Overall, considering all these phenomena would add considerable complexity to our analysis, but no significant insight given our focus on short-run dynamics.

\section{Equilibrium: existence and recursive construction \label{sec:equilibrium}}

In this section, first we provide the definition of an interteporal equilibrium for our economy, which poses particular hidden difficulties (see Section \ref{sub:defequilibrium}) and then provide a theorem of the existence of an equilibrium. Finally, we discuss a recursive construction of equilibrium (see Section \ref{sec:alg}), which is the basis for our numerical simulations.

\subsection{The definition and existence of equilibrium\label{sub:defequilibrium}}

First, we give the definition of a symmetric Nash equilibrium for our Mean Field Game. Let $\mathcal{P}(\K)$ be the set of probability distributions on $\K$, that is $\bm\mu(t) \in \mathcal{P}(\K)$ for every $t \geq  0$.

\begin{Definition}[Symmetric Nash equilibrium of the Mean Field Game]\label{nash}
	Let $\bm\mu (0)\in\mathcal{P}(\K)$ be the health status distribution of a population at $t=0$. A Nash equilibrium for the Mean Field Game is a strategy $\bar{\bm\vartheta}(\cdot,\cdot)\in\mathcal{A}$ such that,
	\begin{equation}\label{nashnash}
	V(0,k(0), \bm{\mu}(0) ,\bar{\bm\vartheta}(\cdot,\cdot)) = J(0,k(0), \bm{\mu}(0),\bar{\bm\vartheta}(\cdot,\cdot);\bar{\bm\vartheta}(\cdot,\cdot)) \qquad  \forall k  \in \K.
	\end{equation}
\end{Definition}
Definition \ref{nash} states that, at equilibrium, the optimal mobility choice of an agent, when the average mobility choice of the other agents is $\bar{\bm\vartheta}(\cdot,\cdot)$, is exactly $\bar{\bm\vartheta}$, that is, the equilibrium is \textit{symmetric} for all agents belonging to the same health status. Focusing on symmetric Nash equilibria among all possible Nash equilibria is very common in the Mean Field literature (see, e.g., \citealp[Sec. 6.1.1.]{carmona}).
	
From another perspective, our Mean Field Game can be viewed as an ''anonymous sequential game with a continuum of players, in which agent players affect their opponents in ways that are insignificant at the individual level but significant when aggregated, and in which factors that are stochastic at the individual level become deterministic when aggregated'' \citep{jovanovic1988anonymous}. In particular, the following notion of equilibrium can be formulated:
\begin{Definition}[Equilibrium of the anonymous sequential game]\label{def:eq}
	An equilibrium starting from $\bm{\mu}(0) \in\mathcal{P}(\K) $ is a couple $(v(\cdot,\cdot),\hat{\bm\vartheta}(\cdot,\cdot))$, with $v:\N\times \K\to\R$ and $\hat{\bm\vartheta}(\cdot,\cdot)\in\mathcal{A}$, such that, along the trajectory of the health status distribution starting at $\bm{\mu}(0)$ as a result of the average strategy $\hat{\bm\vartheta}(\cdot,\cdot)$, one has that:
	\begin{itemize}
		\item[(i)] $v$ is bounded and satisfies\footnote{The trajectory of health status distribution starting at $\bm{\mu}(0)$ enters into \eqref{DPE} by the sequence of $p_{k(t_0)k}$, in turn depending on $\tau(t_0)$ of \eqref{eq:taudef}, that is, the probabilities to change individual health status $p_{k(t_0)k}$ depend on the share of infected agents on population $\mu_I$.} the Bellman equation (\ref{DPE}) for every $(t_{0},k(t_0))\in\N\times \K$;
		\item[(ii)]  $\hat{\bm\vartheta}(t_{0},k(t_0))$ is an optimizer of the right hand side of \eqref{DPE} for every $(t_{0},k(t_0))\in\N\times \K$ when $\bar{\bm\vartheta}(t_{0},k(t_0))=\hat{\bm\vartheta}(t_{0},k(t_0))$.
		\end{itemize}
\end{Definition}

According to Definition \ref{def:eq}, the notion of equilibrium requires that for each $t_0>0$, each agent optimizes their objective functional given its health status $k({t_0})$ and the health status distribution $\bm\mu(t_0)$ (Point (ii) in Definition \ref{def:eq}) and that such optimization is sequentially consistent; that is,  $k(t_{0}+1)$ and $\bm\mu(t_{
0}+1)$ are the outcome of the optimizing behavior at time $t_{0}$; then, $k(t_{0}+2)$ and $\bm\mu(t_{0}+2)$ are the outcome of the optimizing behavior at time $t_{0}+1$; etc. (Point (i) in Definition \ref{def:eq}).

The importance of Definition \ref{def:eq} of equilibrium will be clarified further in Section \ref{sec:alg}, where we will deal with the recursive construction of the equilibrium, the basis of our numerical investigation of the properties of equilibrium.
Notably, the use of Definition \ref{def:eq} in the rest of the analysis is legitimated by its equivalence with Definition \ref{nash}, as proven in Proposition \ref{prop:equivalence}.
\begin{Proposition}\label{prop:equivalence}
	Definitions \ref{nash} and \ref{def:eq} are equivalent.
\end{Proposition}

\begin{proof}
(a)
Let $\bm\mu(0)\in\mathcal{P}(\K)$, let  $(v,\hat{\bm{\vartheta}})$ be an equilibrium in the sense of Definition \ref{def:eq}, and let $k(0)\in  \K$.  
By standard verification arguments in optimal control, it is clear that, since $v$ is bounded, it coincides with the  value function (of the agent) and that the control  $\hat{\bm\vartheta}\in \mathcal{A}$
is optimal (for the agent)  when $\overline{\bm\vartheta}= \hat{\bm\vartheta}$. Hence, \eqref{nashnash} is verified showing that
$\hat{\bm\vartheta}$ is a Nash equilibrium in the sense of Definition \ref{nash}.

(b)
Let $\bm\mu(0)\in\mathcal{P}(\K)$ and let $ \hat{\bm\vartheta}$ be a Nash equilibrium in the sense of Definition \ref{nash}. Set, for each $t_0\geq 0$,
$v(t_0,k(t_0)):=V(t_{0},k(t_0), \bm\mu(t_0),\overline{\bm\vartheta})$ with $\overline{\bm\vartheta}= \hat{\bm\vartheta}$ and consider the couple $(v,\hat{\bm\vartheta})$. By the dynamic programming principle, $v(t_0,k(t_0))$ satisfies \eqref{DPE} at each $t_{0}\geq 0$, so part (i) of Definition \ref{def:eq} is satisfied. Part (ii) of the same definition is satisfied by \eqref{nashnash}.
\end{proof}	
Proposition \ref{prop:equivalence} states that Definitions \ref{nash} and \ref{def:eq} identify the same equilibria, i.e. when our Mean Field Game is viewed as an anonymous sequential game, its equilibrium is a Nash equilibrium and vice versa.

We conclude the section with a result of existence of an equilibrium given in Theorem \ref{th:existence}.
\begin{Theorem}\label{th:existence}
	Given the Definition \ref{def:eq} of the equilibrium of our Mean Field Game, such equilibrium exists for each $\bm\mu(0)\in\mathcal{P}(\K)$.
\end{Theorem}
\begin{proof}
See Appendix \ref{app:proofs}.
\end{proof}
\begin{Remark}
The proof of existence is based on the Tikhonov's fixed point Theorem (see Theorem \ref{th:t} in Appendix \ref{app:proofs}), which however does not guarantee the \textit{uniqueness} of equilibrium. {In MFG theory, the usual condition to guarantee uniqueness are monotonicity conditions on the costs promoting  a disaggregation dynamics (see, e.g., \cite{Cardaliaguet2020}). Even if our model does not fall in the standard families of MFG theory, one can expect that the aggregation push due to the term $Z(t)$ leads to possible multiple equilibria.}  
\end{Remark}

\subsection{The recursive construction of the equilibrium\label{sec:alg}}

In Algorithm \ref{algo:compEquilibrium}, we illustrate a recursive algorithm, inspired by Definition \ref{def:eq}, which allows to compute an equilibrium of our Mean Field Game. The importance of Algorithm \ref{algo:compEquilibrium} is shown by Theorem \ref{thm:verif}, which states the conditions for the computed equilibrium to be both a Nash equilibrium and an anonymous sequential game equilibrium, that is, to satisfy Definitions \ref{nash} and \ref{def:eq}.

\begin{Algorithm}[The algorithm for the computation of an equilibrium] \label{algo:compEquilibrium}.
\begin{enumerate}[1.]
\item At time $0$, set $\hat{\bm\mu}(0)=\bm{\mu}(0)$, $\hat v(0,D)=0$, and arbitrarily assign $\hat{v}(0,k)$ for $k \in \{S,I,R\}$.\\

\item At time $t\geq 0$, given $\hat{\bm\mu}(t)$ and $\hat{v}(t,\cdot)$, according to the corresponding optimization in the Bellman equation (cf. \eqref{eq:DPERbis}-\eqref{eq:DPEI}), we set, for  $k\in\{R,I\}$\footnote{Hereafter, given $a,b\in\R$, we denote $a\vee b=\max\{a,b\}$, $a\wedge b=\min\{a,b\}$.},
\begin{equation}\label{thetaRI}
	\hat {\bm\vartheta}(t,k):=\left(\left(\left(\frac{1}{\gamma_p(t,k,\hat{\bm\mu}(t))}-\frac{A^k_{0}}{A^k_{1}}\right)\vee 0\right)\wedge 1, \ \left(\left(
	\frac{1}{\gamma_c(t,k,\hat{\bm\mu}(t))}-\frac{P_{0}}{P_{1}}\right)\vee 0\right)\wedge 1\right).
	\end{equation}
\item Then, to perform the optimization in the Bellman equation for $k=S$ (cf. \eqref{eq:DPES}), we set
	\begin{equation}\nonumber
	\hat a(t):=\hat{\mu}(t,I)\hat {\vartheta}_{p}(t,I),\ \ \ \  \hat b(t):=\hat{\mu}(t,I) \hat {\vartheta}_{c}(t,I).
	\end{equation}
	and, fixing the difference $\xi:= v(t_{0}+1,S)-v(t_{0}+1,I)$ as a parameter, we set
	 $$
 \hat{\bm\vartheta}^{\xi}(t,S)=(\hat{\vartheta}^{\xi}_{p}(t,S),\hat{\vartheta}^{\xi}_{c}(t,S)),
	$$
	where
	\begin{eqnarray*} \nonumber
		\vartheta_p^{\xi}(t,S)=
		\frac{1}{\gamma_p(t,S,\hat{\bm\mu}(t))+{(1-\rho)}  \hat a(t) \xi}
		-\frac{A^S_{0}}{A^S_{1}}, \\		\vartheta^{\xi}_c(t,S)=
		\frac{1}{\gamma_c(t,S,\hat{\bm\mu}(t))+{(1-\rho)}  \hat b(t)\xi}
		-\frac{P_{0}}{P_{1}}.
   \end{eqnarray*}
Then, \eqref{eq:DPES} can be rewritten in terms of $\xi$ leading to the algebraic equation 	
\begin{equation}\label{eq:HJBS}
	\hat v(t,S) =(1-\rho) \hat v^{\xi}(t+1,I)+(1-\rho) \xi + f(t,\xi),
	\end{equation}
	where:
	\begin{equation}\label{def:f}
		f(t,\xi)= u(t,\hat c^{\xi}(t,S),S,\hat{\bm\vartheta}^{\xi}(t,S))
			- (1-\rho)\left(\beta_{p}\hat a(t) \hat{\vartheta}_p^{\xi}(t,S) +\beta_{c}\hat b(t) \hat{\vartheta}_c^{\xi}(t,S)\right) \xi,
	\end{equation}
\item Given the parametric value $\xi:= v(t_{0}+1,S)-v(t_{0}+1,I)$, we set the value of the corresponding variables:
	\begin{equation} \nonumber
	\begin{array}{l}
	\hat {Z}^{\xi}(t) = \phi\bigg( \hat{\mu}(t,S)\hat\vartheta^{\xi}_p(t,S), \,\,\hat{\mu}(t,I)\hat \vartheta_{p}(t,I), \,\, \hat{\mu}(t,R)\hat\vartheta_{p}(t,R)\bigg); \medskip \\
	\hat c^{\xi}(t,k) = \hat Z^\xi(t) \left( A_{0}^k + A_{1}^k \hat \vartheta_p(t,k)\right) \left(P_{0} + P_{1} \hat\vartheta_c(t,k) \right), \ \ \ \text{for } k=R,I; \medskip \\
	\hat c^{\xi}(t,S) = \hat Z^\xi(t) \left( A_{0}^k + A_{1}^k \hat \vartheta^{\xi}_p(t,S)\right) \left(P_{0} + P_{1}\hat \vartheta^\xi_c(t,S) \right) ; \medskip\\
	\hat v^{\xi}(t+1,R)=\dfrac{1}{1-\rho}\left(\hat{v}(t,R)  -	u(t,\hat c^{\xi}(t,R),R,\hat{\vartheta}(t,R)\right); \medskip\\
	\hat v^{\xi}(t+1,I)=\dfrac{1}{1-\pi_R-\pi_D}\left[\dfrac{\hat v(t,I)-
	u(t,\hat c^{\xi}(t,I),I,\hat{\vartheta}(t,I))}{1-\rho}- \pi_R\, \hat v^{\xi}(t+1,R) \right];\medskip\\
	\hat v^{\xi}(t+1,S)=\xi+ \hat v^{\xi}(t+1,I).	
	\end{array}
	\end{equation}
\item Assuming that \eqref{eq:HJBS} admits a unique solution $\hat\xi$, we set
\begin{equation}\label{thetaS}
		{\hat{\bm\vartheta}}(t,S)= {\hat{\bm\vartheta}}^{\hat\xi}(t,S),
\end{equation}
and the values of the variables at time $t+1$ as
	\begin{equation}\label{def:hatv}
	\begin{cases}
	\hat v(t+1,R)= \hat v^{\hat\xi}(t+1,R),\\
	  \hat v(t+1,I)= \hat v^{\hat\xi}(t+1,I),\\
	   \hat v(t+1,S)= \hat v^{\hat\xi}(t+1,S),
	   \\ \hat v(t+1,D)= 0,
	   \end{cases}
\end{equation}
	and
	$$
	\hat{\bm{\mu}}(t+1)=\hat {\bm{Q}}(t)\hat{\bm{\mu}}(t),
	$$
	where
	$$
	\bm{\hat Q}(t):=\left(
	\begin{array}{cccc}
	1- \hat \beta(t)\hat \mu(t,I) & 0 & 0 & 0 \\
	\hat \beta(t)\hat \mu(t,I)  & 1- \pi_R -\pi_D  & 0 & 0 \\
	0 & \pi_R & 1 & 0  \\
	0 & \pi_D & 0 & 1 \\
	\end{array}
	\right),
	$$
	where
	$$
	\hat \beta(t) :=\beta_{p}\hat\vartheta_{p}(t,I){\hat\vartheta}_p(t,S) +\beta_{c}\hat\vartheta_{c}(t,I){\hat\vartheta}_c(t,S).
	$$
	
	\item We repeat steps 2-4 with the updated $\hat{\bm\mu}(t+1)$ and $\hat{v}(t+1,\cdot)$.	
\end{enumerate}
\end{Algorithm}

\begin{Theorem}\label{thm:verif}
	Let $\bm{\mu}(0)$ is the initial health status distribution and let $\hat{v}(0,\cdot)$ be assigned with $\hat v(0,D)=0$. Consider Algorithm \ref{algo:compEquilibrium} and assume that $\hat\xi$ is well defined for every $t\in\N$ and that $\hat{v}$ is bounded. Then the couple $(\hat{v},\hat{\bm\vartheta})$ is an equilibrium starting at $\bm{\mu}(0)$ according to Definition \ref{def:eq}.
\end{Theorem}
\begin{proof}
See Appendix \ref{app:proofs}.
\end{proof}

The logic behind the use of Algorithm \ref{algo:compEquilibrium} together with Theorem \ref{thm:verif} is that the search for the equilibrium of our Mean Field Game can be traced back to the search for the initial value  $v(0,\cdot)$ such that the implied dynamics of $v(t,\cdot)$, starting from the initial health status distribution $\bm{\mu}(0)$, is consistent with the optimal conditions and $v(t,\cdot)$ is both non negative (we have normalized $v(t,D)=0$ for each $t$ by an appropriate choice of $M$) and bounded from above.

\section{Calibration of the model\label{sec:calibrationValidationModel}}

In the calibration of the model, we focus on the recent Italian experience with COVID-19. Italy was unfortunately the first Western country severely hit by COVID-19; the epidemic shock was sudden and unexpected as well as the deep impact on Italian mobility and production (see Figure \ref{fig:weeklymobilityforworkplacesvseconomicactivity} below). At the same time, Italy was also the first Western country to adopt strict restrictions in mobility in March 2020. Overall, this makes the Italian case particularly well-adapted to calibrate/estimate the relationship between mobility, production and dynamics of epidemic.\footnote{Data and codes are available at \url{https://people.unipi.it/davide_fiaschi/ricerca/}.}

The first step in the numerical calibration of the model is to specify the $Z(t)$ in (\ref{eq:agentIncome}). To minimize the number of model's parameters, we consider the following one-parameter specification:
\begin{equation}
Z(t) \equiv  1 - \exp\left( - g \left[ \bar\vartheta_p(t,S) \mu(t,S) + \bar\vartheta_p(t,I) \mu(t,I) + \bar\vartheta_p(t,R) \mu(t,R) \right] \right),
\label{eq:functionZ}
\end{equation}
where $g$ measures the sensitivity of individual income to aggregate mobility, i.e. the complementarities between individual and aggregate mobility in determining the level of individual income. In this respect, we expect that $g$ is greater than 0. Taking (\ref{eq:functionZ}) into account, overall we have to set 19 parameters, which are listed in Table \ref{tab:listModelParameters}. Below, we provide more details on the method used to set their values.

\begin{table}[htbp]
	\scriptsize{
		\begin{tabular}{p{2.cm}p{5cm}p{3cm}p{7.cm}}
			\hline
			\hline
			Parameter & Meaning & Value & Method used to set the value \\
			\hline
			$\pi_R$ & Daily probability of recovering when infected & 0.07143 & Taken from literature on COVID-19 \citep{voinsky2020effects}\\
			$\pi_D$ & Daily probability of death when infected & 0.00052 & Taken from literature on COVID-19 \citep{flaxman2020estimating}\\
			$\beta_p$ & The impact of mobility for production on infection & 0.14902 & Calculated based on an $R_0$ equal to 2.9 for Italy (\url{https://en.wikipedia.org/wiki/Basic\_reproduction\_number}) and on the fact that mobility of infected is on average 30\% less as a result of prevalent rate of symptoms of COVID-19 in infected people \citep{day2020covid}\\
			$\beta_c$ & The impact of mobility for consumption on infection & 0.14606 & Calculated based on an $R_0$ equal to 2.9 for Italy (\url{https://en.wikipedia.org/wiki/Basic\_reproduction\_number}) and on the fact that mobility of infected is on average 30\% less as a result of prevalent rate of symptoms of COVID-19 in infected people \citep{day2020covid} \\
			$\rho$ & Discount rate of utilities & 0.000296 & Taken from \cite{Laibsonetal18} \\
			$\gamma_p(S)$, $\gamma_p(I)$, and $\gamma_p(R)$ & Cost of mobility for production for different types of agents in baseline scenario & 0.29795, 0.42564, and 0.29795 & Calibrated in order to have mobility and production equal to 1 in a free-epidemic economy for susceptibles and recovered and mobility equal to 0.7 for infected\\
			$\gamma_c(S)$, $\gamma_c(I)$, and $\gamma_c(R)$ & Cost of mobility for consumption for different types of agents in baseline scenario & 0.21375, 0.22840, and 0.21375 & Calibrated in order to have mobility and production equal to 1 in a free-epidemic economy for susceptibles and recovered and mobility equal to 0.7 for infected\\
			$A_0^{SR}$ and $A_0^{I}$ & Sensibility of individual income to aggregate mobility independent from individual mobility & 0.70229 and 0.49160 & For susceptible and recovered estimated from the relation between mobility and production in Italy in the period February 2020 - May 2021 (see Figure \ref{fig:weeklymobilityforworkplacesvseconomicactivity}). For infected people calibrated at 70\% of other agents based on the prevalence of symptoms.\\
			$A_1^{SR}$ and $A_1^{I}$ & Sensibility of individual income to individual mobility & 0.29805 and 0.29805 & Estimated from the relation between mobility and production in Italy in the period February 2020 - May 2021 setting mobility and production equal to 1 in a pre-epidemic economy (see Figure \ref{fig:weeklymobilityforworkplacesvseconomicactivity})\\
			$P_0$ and $P_1$ & Sensibility of individual consumption to individual mobility & 0.47187 and 0.12828 & Estimated from the relation between average propensity to consume and mobility for retail and recreation in Italy in the period February 2020 - May 2021\\
			$g$ & Sensibility of individual income to aggregate mobility & $ 7.741615 $ & Estimated from the relation betweem mobility and production in Italy in the period February 2020 - May 2021 (see Figure \ref{fig:weeklymobilityforworkplacesvseconomicactivity}) \\
			$M$ & Utility to be deceased & -1.30 & Calibrated to avoid negative lifetime utility for each survival agent\\
			$\mu(0,S)$, $\mu(0,I)$, and $\mu(0,R)$ & Initial state of epidemic & $1 -1/60.000.000$,
			$1/60.000.000$, and 0 & Calibrate on an economy of 60 million agents as in Italy in 2020\\
			\hline
			\hline
		\end{tabular}
	}
	\caption{List of model's parameters, their values and notes on how they are calculated/calibrated/estimated.}
	\label{tab:listModelParameters}
\end{table}

\subsection{Calibration of the epidemiological parameters}

The calibration of the epidemiological parameters focuses on daily dynamics as standard in epidemiology \citep{ferguson2020report}. Several studies provide basic information on COVID-19 main epidemiological characteristics. In particular, \cite{voinsky2020effects} report that the average number of days for recovering from COVID-19 is 14, which implies $\pi_R =0.07142$. \cite{flaxman2020estimating}, instead, document an overall probability to die once infected of 0.94\% in Italy and an average number of days from infection to death of 18, which implies $\pi_D=0.00052$.

Finally, for setting $\beta_p$ and $\beta_c$ we assume that they are equal, so that observed infection rate is the product between $\beta_p$ ($\beta_c$) and the average mobility of infected agents once mobility of susceptible is normalized to one in an economy without infected, that is, $\bar{\bm \vartheta}(0,S)=(1,1)$ (see System of (\ref{eq:evolu})). \cite{day2020covid} report that the prevalence rate of symptoms of COVID-19 in infected people is about 30\%, i.e. 70\% of infected people are asymptomatic. Assuming that the latter maintain the same mobility, we set average mobility of an infected agent 30\% less than the one of a susceptible, that is, $\bar{\bm \vartheta}(0,I)=(0.7,0.7)$. Since the observed infection rate at time 0 can be expressed as $\beta(0) = \left(\pi_{D} + \pi_{R}\right)R_0$, then $\beta(0) = \beta_{p}\bar{\vartheta}_p(0,S)\bar{\vartheta}_p(0,I) + \beta_{c}\bar{\vartheta}_c(0,S)\bar{\vartheta}_c(0,I) =\left(\pi_{D} + \pi_{R}\right)R_0$, therefore $2 \beta_{p}\bar{\vartheta}_p(0,I) =\left(\pi_{D} + \pi_{R}\right)R_0$, and, finally, $\beta_p=\beta_c = (1/1.4)\left(\pi_{D} + \pi_{R}\right)R_0=0.14902$, given a basic reproduction rate $R_0$ of COVID-19 equal to 2.9 for Italy.\footnote{\url{https://en.wikipedia.org/wiki/Basic\_reproduction\_number}.}

\subsection{Calibration of the economic part}

The calibration of parameters governing the relationship between income and mobility are based on the Italian experience in the period February 15, 2020 - May 31, 2021 reported in Figure \ref{fig:weeklymobilityforworkplacesvseconomicactivity}.

\begin{figure}[htbp]
	\centering
	\includegraphics[width=0.9\linewidth]{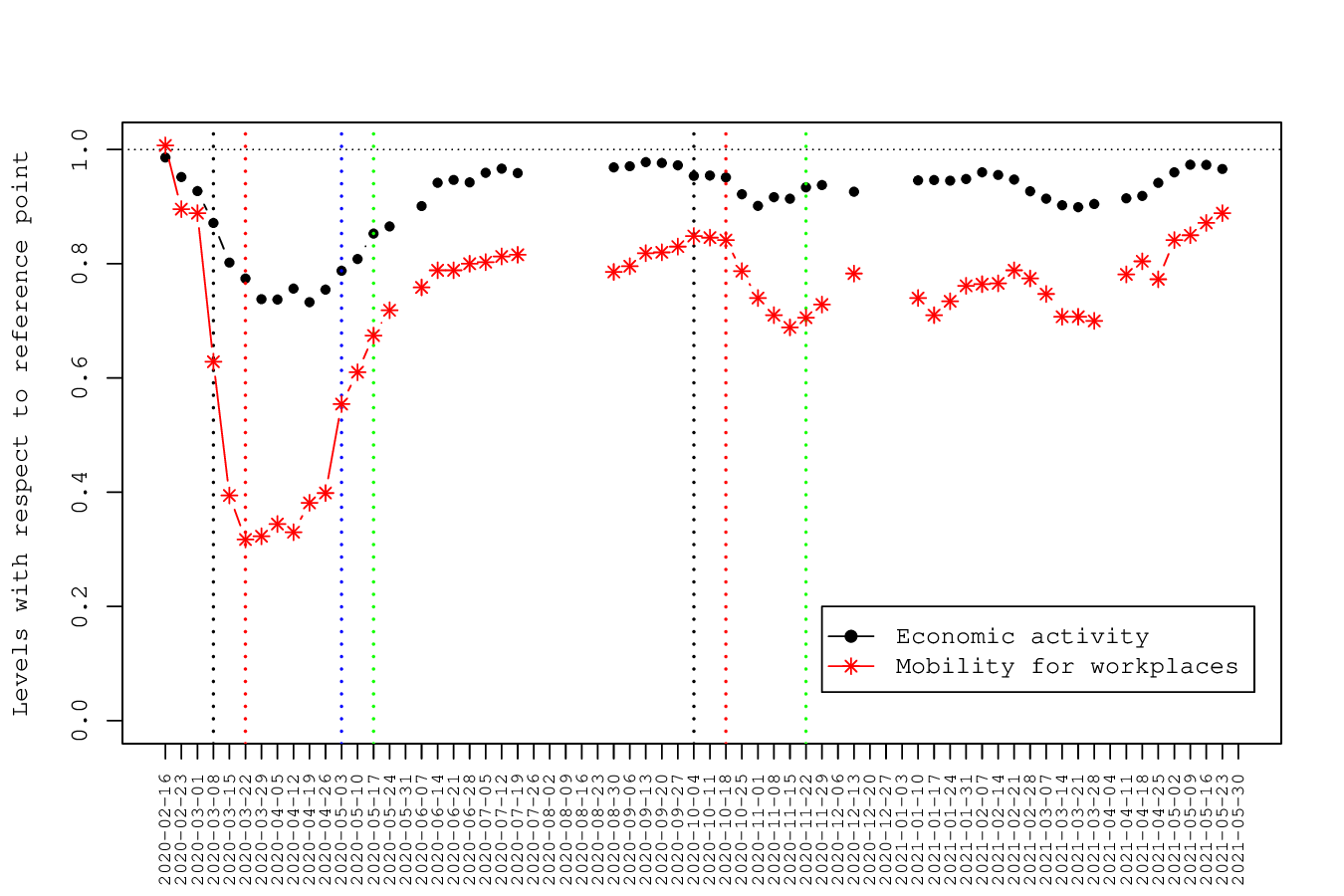}
	\caption{The relationship between weakly mobility for workplace and weakly economic activity in the period February 15, 2020 - May 31, 2021 (Italian holiday weeks are not reported). Dashed lines indicate weeks of new imposed mobility restrictions at national level (March 9, 2020, March 22, 2020, October 8, 2020 and October 24, 2020) and of a relaxation in mobility restrictions (May 4, 2020, May 18, 2020, and November 24, 2020). \textit{Source}:  Google Mobility Trend (\url{https://www.google.com/covid19/mobility/}) and OECD Weekly Tracker of GDP growth (\url{https://www.oecd.org/economy/weekly-tracker-of-gdp-growth/}) }
	\label{fig:weeklymobilityforworkplacesvseconomicactivity}
\end{figure}

Italian economic activity as estimated by OECD Weekly Tracker of GDP growth\footnote{\url{https://www.oecd.org/economy/weekly-tracker-of-gdp-growth/}.} appears very correlated with mobility for workplaces as reported by the Google Mobility Trend. \footnote{\url{https://www.google.com/covid19/mobility/}).} The strong drop in mobility in the period between February 23, 2020 and March 8, 2020 (almost - 10\%) well before the first introduction of mobility restrictions at national level in the week of March 8, 2020, supports our idea of an endogenous response of agent to epidemic evolution, which burst in Italy at the end of February 2020. The severe restrictions on mobility imposed in two steps in March 2020 led to a drop in mobility and economic activity of about 70\% and 25\% with respect to reference period, respectively. The relaxed restrictions in May 2020 led to a bounce back in both variables, but recovery was not complete. In the autumn of 2020, as a result of the second pandemic wave, Italy again experienced new mobility restrictions, with associated reduction in economic activity.

Normalizing economic activity and mobility to 1 in an economy with only susceptible, and taking (\ref{eq:agentIncome}) and (\ref{eq:functionZ}) to formulate a (nonlinear) relationship between mobility and economic activity, a nonlinear estimation procedure produces an estimate of $g$, $A_0^{SR}$ and $A_1^{SR}$ of 0.70229, 0.29805 and 7.74162, respectively. $A_0^{I}$ and $A_1^{I}$ are set to 0.49160 and 0.29805 to accommodate the assumption that mobility of an infected agent is 70\% of the susceptible one.

As regards to $P_0$ and $P_1$, they are set to indicate that, according to (\ref{eq:agentIncome}) and (\ref{eq:consumtption}), average propensity to consume can be expressed as a function of consumption mobility, $P_0$, and $P_1$. Taking the mobility for retail and recreation from Google Mobility Trend\footnote{\url{https://www.google.com/covid19/mobility/}.} as a proxy for consumption mobility, and the quarterly average propensity to consume from Italian national accounts, we estimate $P_0=0.47187$ and $P_1=0.12828$.
Finally, the utility of state deceased $M$ is set equal to $-1.3$ to avoid that, independent of state of epidemic and economic activity, lifetime utility of survival agents can be negative.

\subsection{SIRD versus economic SIRD (ESIRD) model}

\begin{table}[!htbp] \centering
	\caption{SIRD versus economic SIRD (ESIRD) model with endogenous mobility. Numerical experiments based on the parameters reported in Table \ref{tab:listModelParameters}.}
	\label{tab:dumbVsBaseline}
	\scriptsize{
		\begin{tabular}{@{\extracolsep{5pt}}p{1.8cm}p{1.3cm}p{1.3cm}p{1.3cm}p{1.2cm}p{1cm}p{0.9cm}p{1cm}p{1cm}p{1cm}p{1.cm}}
			\\[-1.8ex]\hline
			\hline \\[-1.8ex]
			Model & Peak prevalence & Cumulative deaths & Minimum of production & Minimum of mobility & Economic loss & Mobility loss & $\mu(425,S)$ & $\mu(425,I)$ & $\mu(425,R)$ & $\mu(425,D)$ (death rate) \\
			\hline \\[-1.8ex]
			SIRD & $17,784,284$ & $408,678$ & $0.87$ & $0.79$ & $-0.011$ & $-0.019 $ & $0.062 $ & $0.000 $ & $0.932$ & $0.007  $ \\
			ESIRD & $5,858,062$ & $297,577$ & $0.883$ & $0.693$ & $-0.032$ & $-0.082$ & $0.314$ & $0.003$ & $0.678$ & $0.005$ \\
			\hline \\[-1.8ex]
		\end{tabular}
	}
\end{table}

Table \ref{tab:dumbVsBaseline} and Figure \ref{fig:comparisonDumbVsEndoMobilitySIR}
highlight the importance of considering endogenous mobility choice in the analysis. In particular, the comparison between the ``na\"ive'' SIRD (where mobility of susceptible, infected and recovered is maintained constant for the whole period of simulation and equal to their initial baseline values), and the ESIRD model (where individual mobility is decided in an optimizing framework without any imposed restriction), points out the 30\% more cumulative deaths of na\"ive SIR as opposed to a lower drop in mobility and production (both as peak and as cumulative impact). After 425 days from its outbreak, the epidemic is substantially ended in both models, that is, $\mu(I)$ is almost zero, but the optimized mobility of an agent in ESIRD has led to a non-negligible mass of susceptibles equal to 31.4\% in day 425 and substantially lower death rate (0.5\% versus 0.7\%).

\begin{figure}[htbp]
	\centering
	\begin{subfigure}[h]{0.49\textwidth}
		\centering
		\includegraphics[width=1\linewidth]{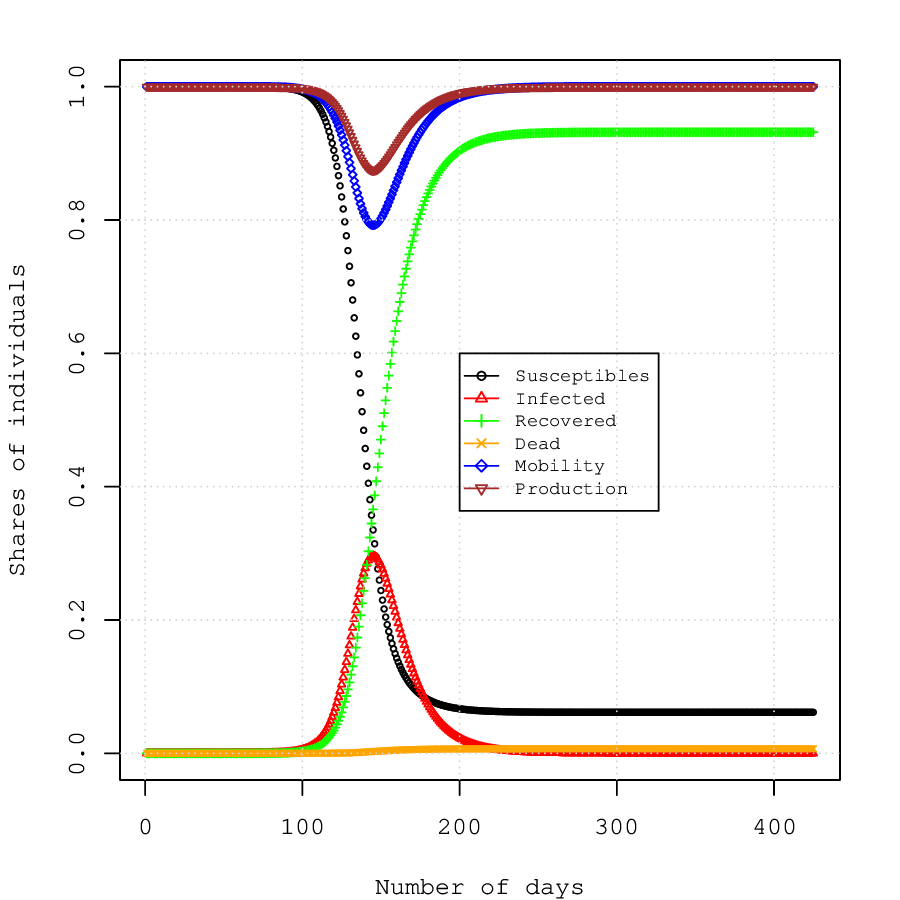}
		\caption{Dynamics of epidemic, economic activity and mobility with ``na\"ive'' agents}
		\label{fig:simSIRModelItalyWithMobilityProduction}
	\end{subfigure}
	\begin{subfigure}[h]{0.49\textwidth}
		\centering
		\includegraphics[width=1\linewidth]{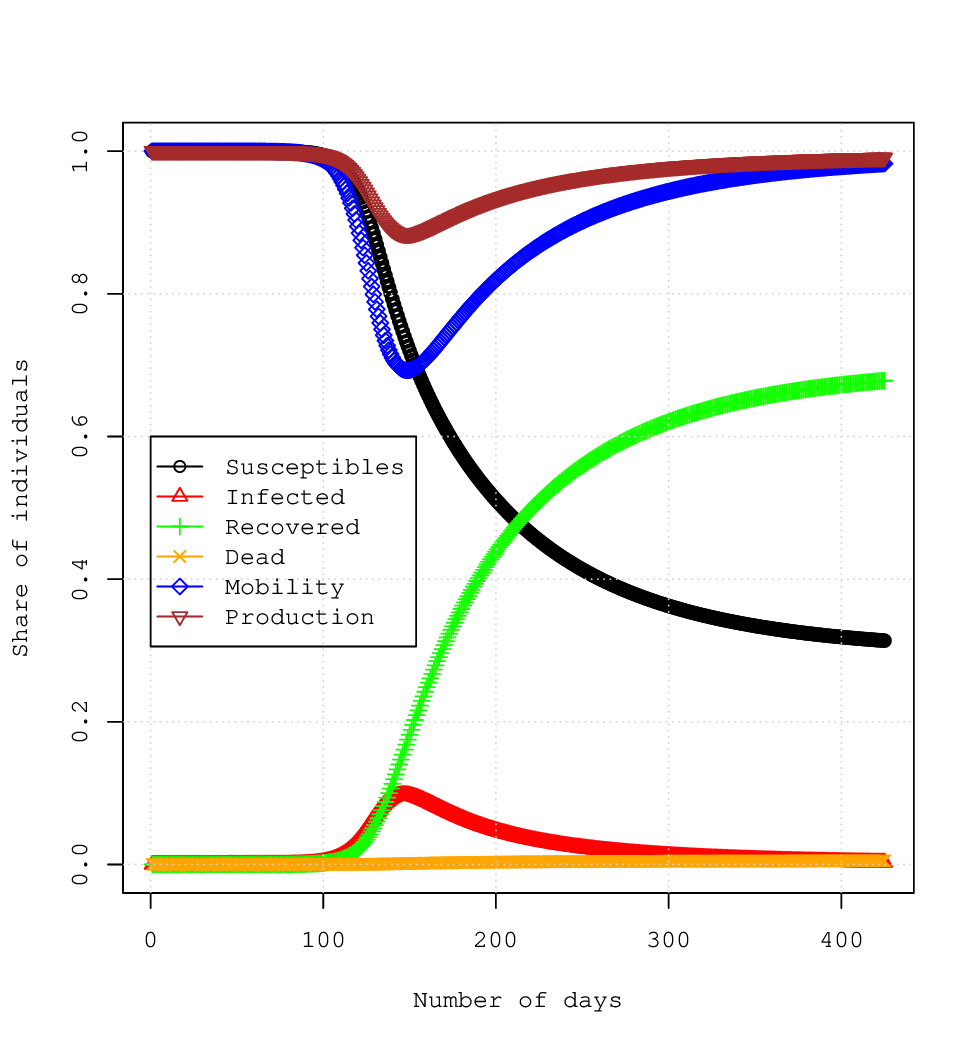}
		\caption{Dynamics of epidemic, economic activity and mobility with agents optimizing their mobility choices.}
		\label{fig:simMacroSIRModelItalyBaselineScenarios}
	\end{subfigure}
	\caption{Comparison between ``na\"ive'' SIRD model versus SIRD model with endogenous mobility. Numerical experiments based on the parameters reported in Table \ref{tab:listModelParameters}.}
	\label{fig:comparisonDumbVsEndoMobilitySIR}
\end{figure}

\section{Questioning the ESIRD\label{sec:questionModel}}

In this section, we discuss how our framework could be used to evaluate alternative policies of mobility restriction.
The high peak prevalence reported for ESIRD in Table \ref{tab:dumbVsBaseline} explains why several countries imposed strong mobility restrictions in 2020. A peak of infected of 5,858,062 agents would correspond to a need of about 398,749 beds in hospitals, taking 6.8\% the proportion of infected
individuals hospitalised \citep{verity2020estimates}. For example, Italy in February 2020 had about 190,000 available beds in hospital, making ''laissez faire'' approach to COVID-19 not practical (not considering the advantage to take time in waiting for a vaccine).

In the following, we therefore study some mitigation strategies as defined in \cite{ferguson2020report} (page 3), that is, ''to use NPIs (non-pharmaceutical intervention) not to interrupt transmission completely, but to reduce the health impact of an epidemic'' in the hope (as it effectively happened) of a rapid development of a vaccine. We will focus on policies that, by increasing mobility costs ($\gamma$s), reduce individual mobility and therefore the infection rate and the peak prevalence. In this regard, \cite{nouvellet2020report} provides strong evidence that reducing mobility is the key factor for bringing down COVID-19 transmission, while \cite{vollmer2020report} present scenario analysis based on different mobility in Italy.

At the same time, reducing mobility negatively impacts production, putting policy makers before a trade-off between economic losses and fatalities due to COVID-19, that is, it is possible to point out a \textit{pandemic possibilities frontier} as in \cite{kaplan2020great} and \cite{acemoglu2020multi}. However, we add two dimensions in the discussion. The first is related to the share of remaining susceptible at the end of the period of analysis, which could facilitate a new outbreak of epidemic in the future. The second related to the social feasibility of some policies based on a long reduction of individual mobility.

Table \ref{tab:alternativeScenarios} reports the effect of different policies increasing (in the same percentage) the cost of mobility for production and consumption with respect to the baseline model when the share of infected individuals exceeds 3\% and to maintain this increase until the share of infected individuals gets down to 0.5\% or to 0.1\% in the more severe scenario (mrs).

\begin{table}[!htbp] \centering
	\caption{Alternative scenarios of restriction of mobility (severity of lockdown) and exit from these restrictions (mrs adopts a more strict threshold for relaxing the restrictions). Numerical experiments based on the parameters reported in Table \ref{tab:listModelParameters}.}
	\label{tab:alternativeScenarios}
	\scriptsize{
		\begin{tabular}{@{\extracolsep{5pt}}p{2.5cm}p{1.3cm}p{1.3cm}p{1.1cm}p{1.cm}p{1cm}p{0.9cm}p{1cm}p{1cm}p{1cm}p{1.cm}}
			\\[-1.8ex]\hline
			\hline \\[-1.8ex]
			Scenario & Peak prevalence & Cumulative deaths & Minimum of production & Minimum of mobility & Economic loss & Mobility loss & $\mu(425,S)$ & $\mu(425,I)$ & $\mu(425,R)$ & $\mu(425,D)$ (death rate) \\
			\hline \\[-1.8ex]
			Baseline ESIRD & $5,858,062$ & $297,577$ & $0.883$ & $0.693$ & $-0.032$ & $-0.082$ & $0.314$ & $0.003$ & $0.678$ & $0.005$ \\
			Cost +10\% & $3,594,938$ & $248,258$ & $0.877$ & $0.651$ & $-0.056$ & $-0.165$ & $0.424$ & $0.006$ & $0.566$ & $0.004$ \\
			Cost +20\% & $1,633,960$ & $160,311$ & $0.867$ & $0.603$ & $-0.083$ & $-0.254$ & $0.626$ & $0.005$ & $0.365$ & $0.003$ \\
			Cost +30\% & $1,275,206$ & $107,837$ & $0.837$ & $0.518$ & $-0.092$ & $-0.280$ & $0.732$ & $0.020$ & $0.246$ & $0.002$ \\
			Cost +40\% & $1,258,593$ & $113,914$ & $0.800$ & $0.439$ & $-0.103$ & $-0.299$ & $0.729$ & $0.010$ & $0.260$ & $0.002$ \\
			Cost +50\% & $1,249,959$ & $111,359$ & $0.753$ & $0.357$ & $-0.113$ & $-0.310$ & $0.733$ & $0.011$ & $0.254$ & $0.002$ \\
			Cost +30\% (mrs) & $1,241,037$ & $75,794$ & $0.835$ & $0.515$ & $-0.100$ & $-0.307$ & $0.824$ & $0.002$ & $0.173$ & $0.001$ \\
			Cost +50\% (mrs) & $1,256,080$ & $67,485$ & $0.747$ & $0.348$ & $-0.122$ & $-0.335$ & $0.841$ & $0.004$ & $0.154$ & $0.001$ \\
			\hline \\[-1.8ex]
		\end{tabular}
	}
\end{table}

Peak prevalence decreases up to a rise of 30\% in mobility cost and then it is almost rigid to further increment (see Table \ref{tab:alternativeScenarios}). Peak prevalence of 1,275,206 individuals would amount to a need of 86,801 beds in hospitals. Non-reported numerical investigations show that to decrease this peak prevalence would require to start mobility restrictions with a lower share of infected individuals than 3\%.

However, increasing mobility costs have also a growing negative impact both on economic activity and on the death rate. This trade-off is represented in Figure \ref{fig:tradeOffIncomeMortalityRateLockdown}, which corresponds to the pandemic possibilities frontier discussed in \cite{kaplan2020great} and \cite{acemoglu2020multi}, but calculated in a very different theoretical framework. We can appreciate from Figure \ref{fig:tradeOffIncomeMortalityRateLockdown} how a scenario with 30\% of additional costs and an exit threshold of 0.1\% from mobility restriction Pareto dominates the scenarios both with 40\% and 50\% of additional costs and an exit threshold of 0.5\%.

\begin{figure}[htbp]
	\centering
	\begin{subfigure}[h]{0.48\textwidth}
		\centering
		\includegraphics[width=1\linewidth]{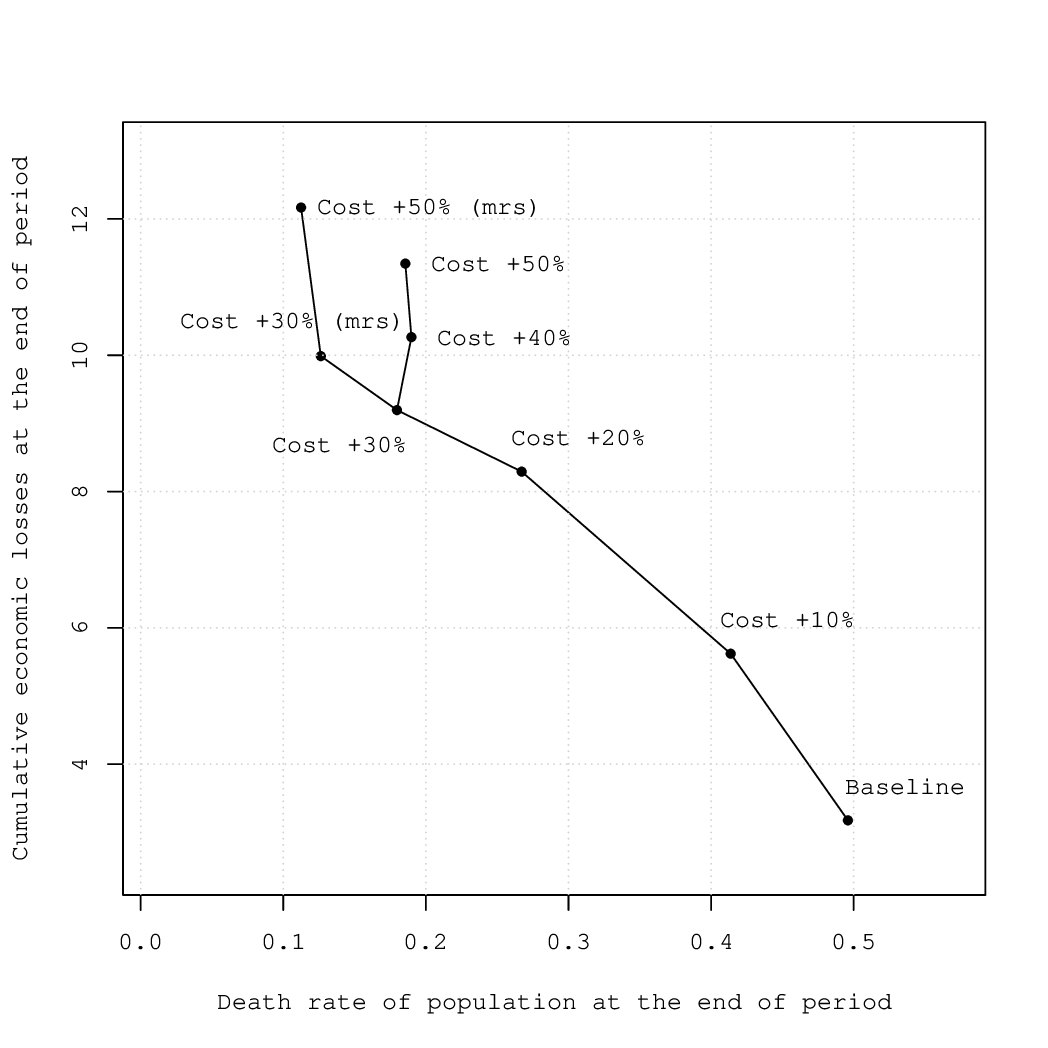}
		\caption{Trade-off between cumulative economic losses and cumulative death rates (number of fatalities on total population) after 425 days from the outbreak of the epidemic in different scenarios}
		\label{fig:tradeOffIncomeMortalityRateLockdown}
	\end{subfigure}
	\begin{subfigure}[h]{0.48\textwidth}
		\centering
		\includegraphics[width=1\linewidth]{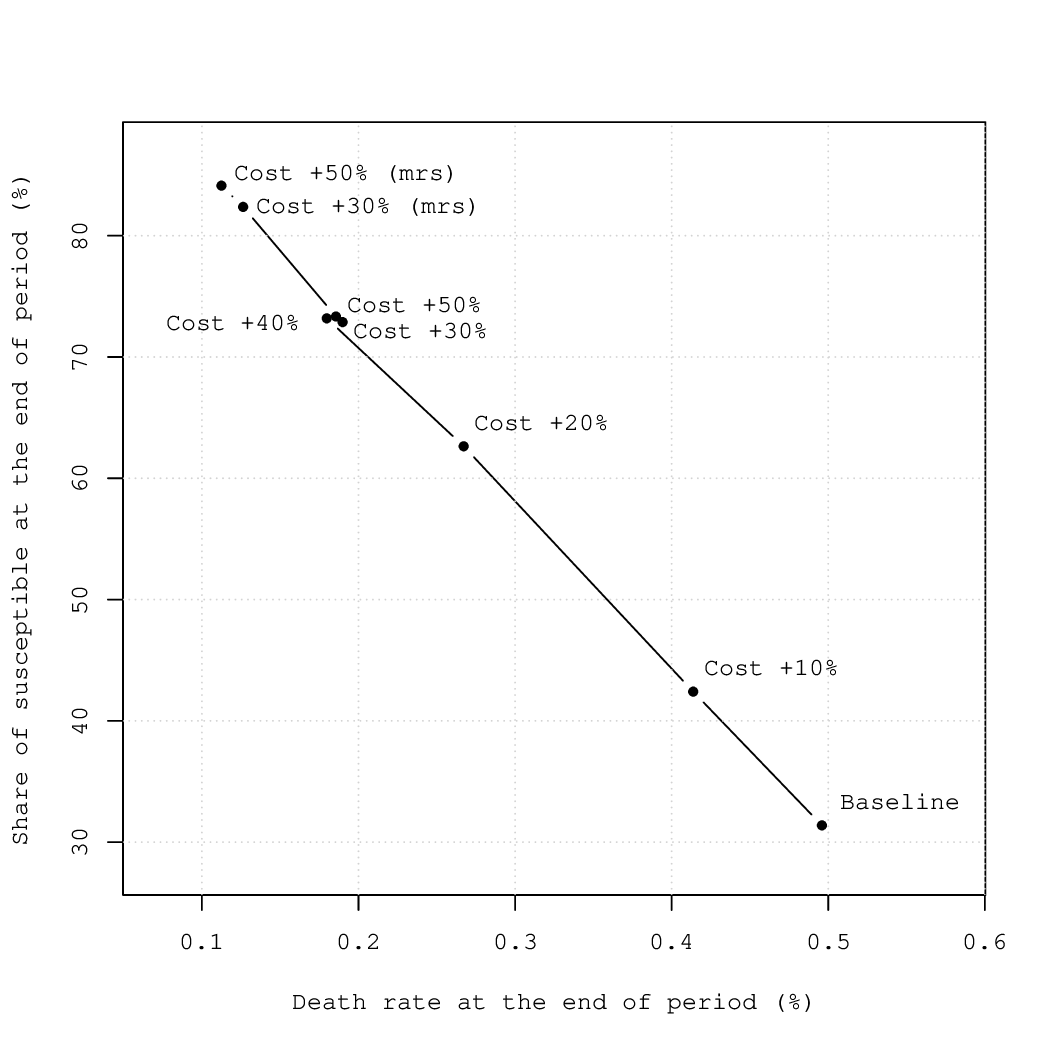}
		\caption{Trade-off between cumulative death rates (number of fatalities on total population) and the share of susceptibles after 425 days from the outbreak of the epidemic in different scenarios}
		\label{fig:tradeOffMoralityRateSusceptibleLockdown}
	\end{subfigure}
	\caption{Trade-offs in alternative scenarios of mobility restrictions and exit from these restrictions. Numerical experiments based on the parameters reported in Table \ref{tab:listModelParameters}.}
	\label{fig:tradeOffAlternativeScenarios}	
\end{figure}

\begin{figure}[htbp]
	\centering
	\begin{subfigure}[h]{0.32\textwidth}
		\centering
		\includegraphics[width=1\linewidth]{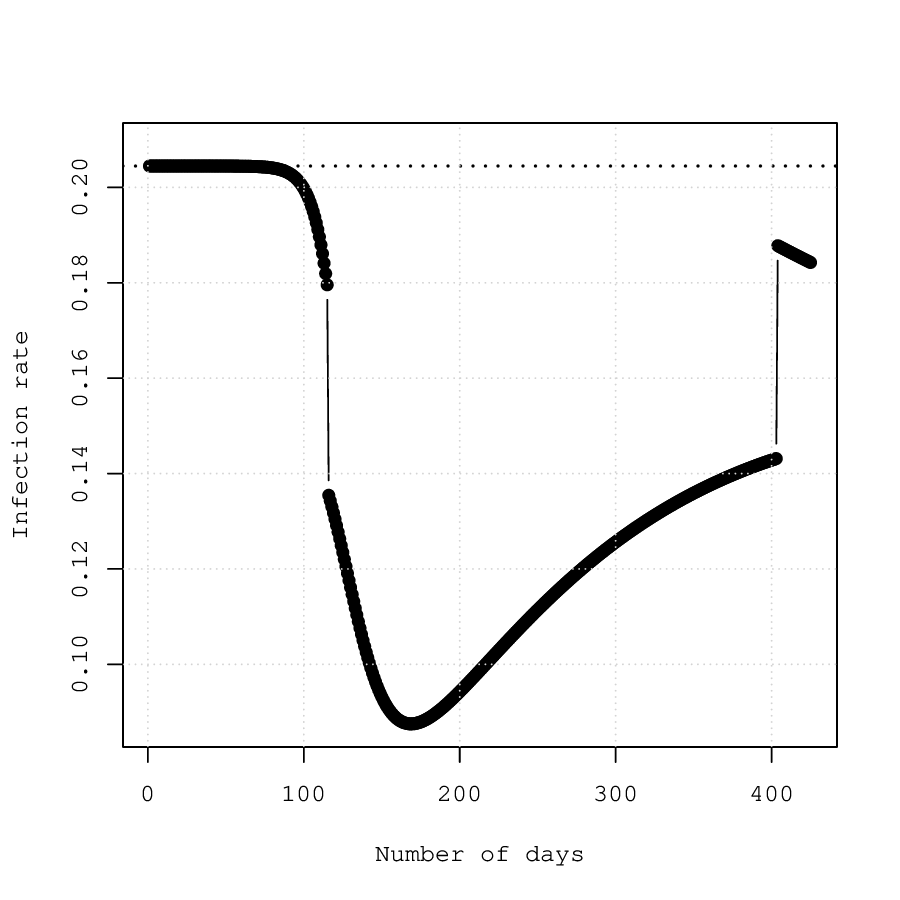}
		\caption{Infection rate when lockdown implies an increase of 10\% of cost of mobility}
		\label{fig:infectionRateItalyBaselineScenariosAddCost0.1}
	\end{subfigure}
	\begin{subfigure}[h]{0.32\textwidth}
		\centering
		\includegraphics[width=1\linewidth]{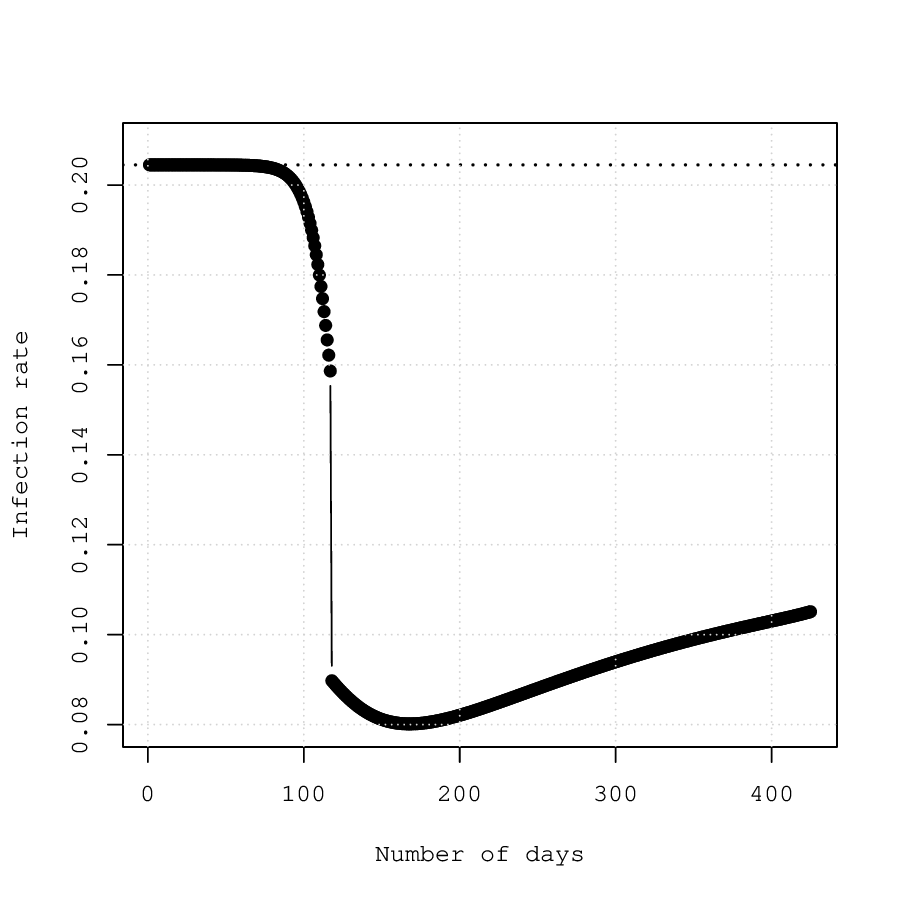}
		\caption{Infection rate when lockdown implies an increase of 20\% of cost of mobility}
		\label{fig:infectionRateItalyBaselineScenariosAddCost0.2}
	\end{subfigure}
	\begin{subfigure}[h]{0.32\textwidth}
		\centering
		\includegraphics[width=1\linewidth]{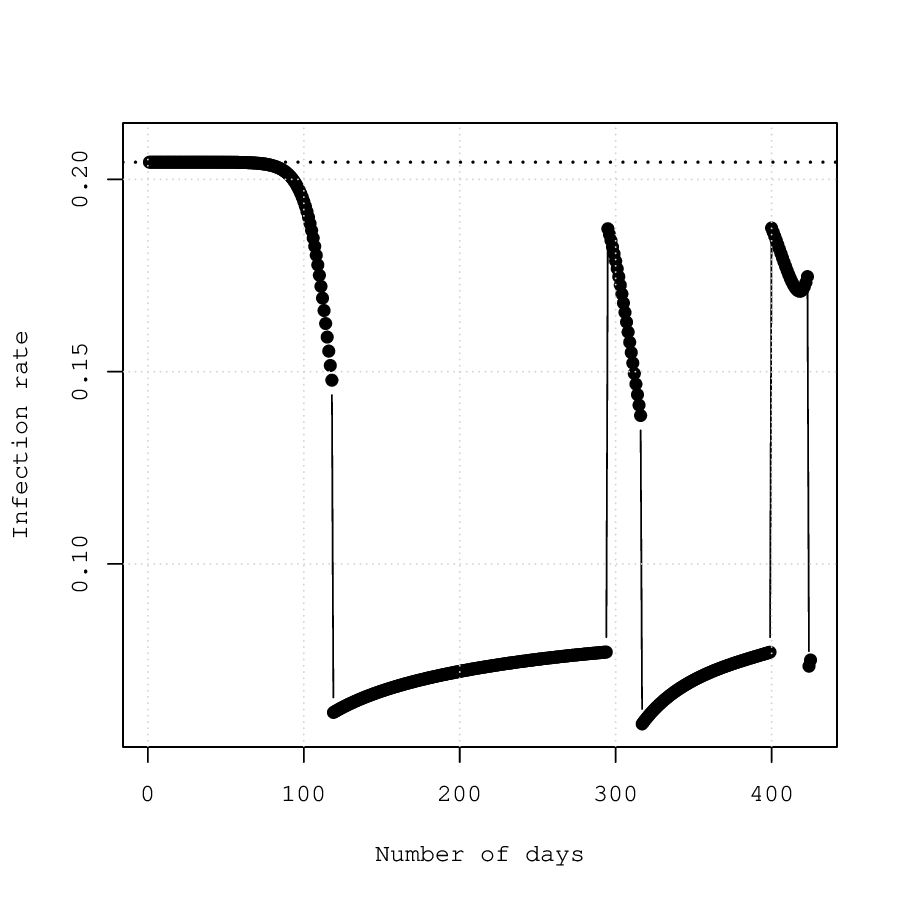}
		\caption{Infection rate when lockdown implies an increase of 30\% of cost of mobility}
		\label{fig:infectionRateItalyBaselineScenariosAddCost0.3}
	\end{subfigure}
	\begin{subfigure}[h]{0.32\textwidth}
		\centering
		\includegraphics[width=1\linewidth]{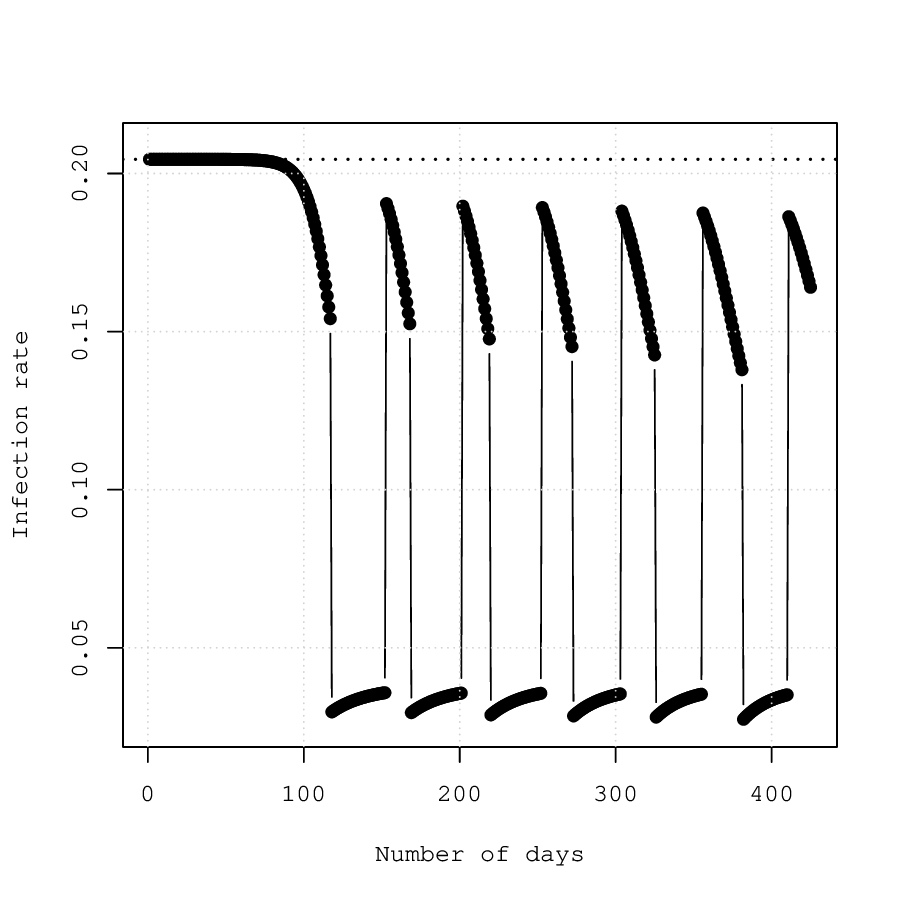}
		\caption{Infection rate when lockdown implies an increase of 50\% of cost of mobility}
		\label{fig:infectionRateItalyBaselineScenariosAddCost0.5}
	\end{subfigure}
	\begin{subfigure}[h]{0.32\textwidth}
		\centering
		\includegraphics[width=1\linewidth]{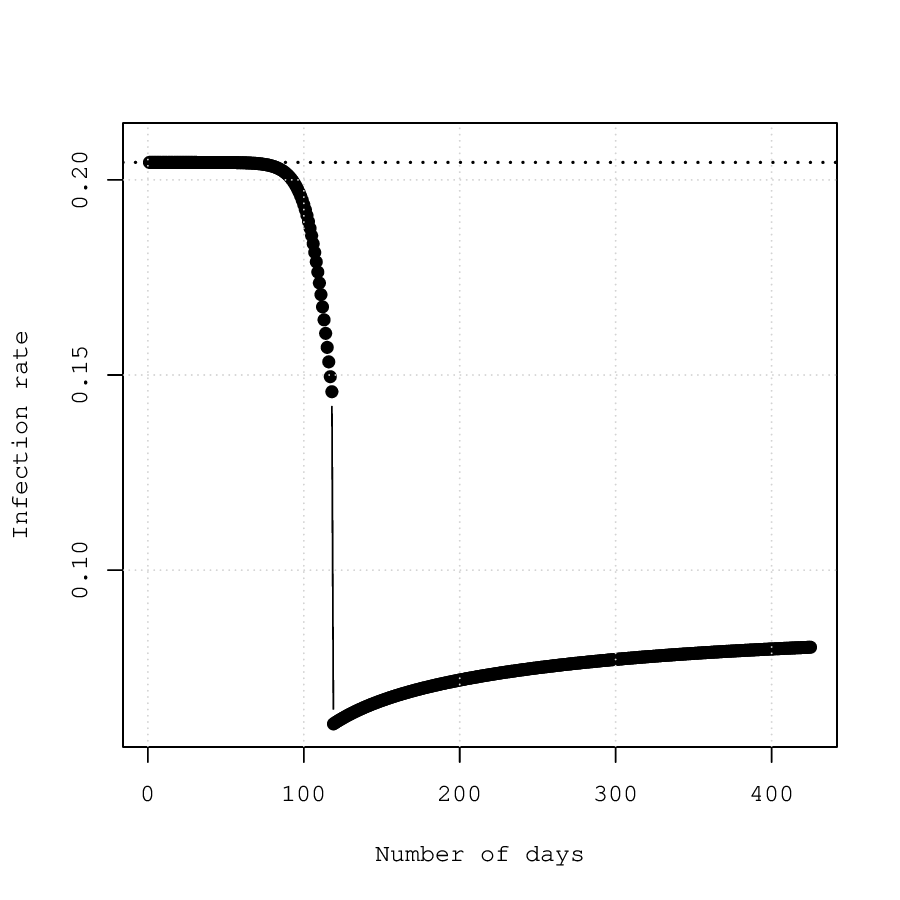}
		\caption{Infection rate when lockdown implies an increase of 30\% of cost of mobility and more restrictive conditions for the exit of lockdown}
		\label{fig:infectionRateItalyBaselineScenariosAddCost0.3exit0.001}
	\end{subfigure}
	\begin{subfigure}[h]{0.32\textwidth}
		\centering
		\includegraphics[width=1\linewidth]{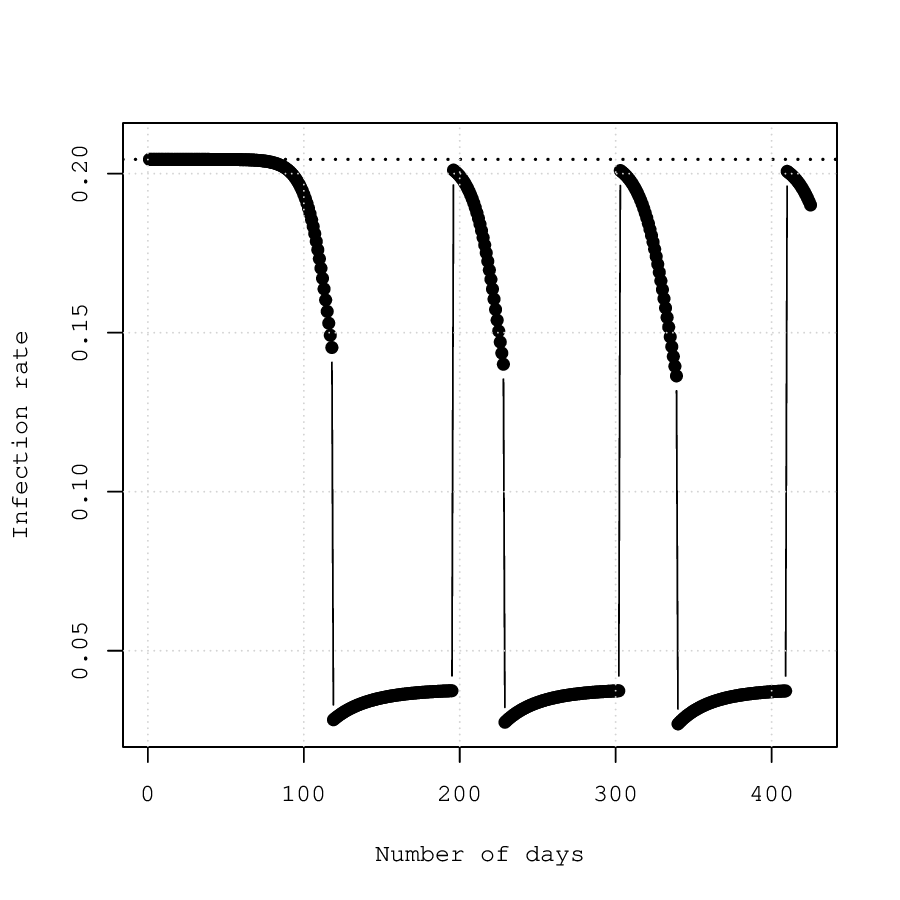}
		\caption{Infection rate when lockdown implies an increase of 50\% of cost of mobility and more restrictive conditions for the exit of lockdown}
		\label{fig:infectionRateItalyBaselineScenariosAddCost0.5exit0.001}
	\end{subfigure}
	\caption{Dynamics of infection rate in different scenarios of mobility restriction (severity of lockdown) and exit from these restrictions (threshold for relaxing the restrictions). Numerical experiments based on the parameters reported in Table \ref{tab:listModelParameters}.}
	\label{fig:infectionRateDifferentScenarios}	
	\end{figure}

However, the former scenario presents two additional non-favorable characteristics with respect to the latter. First, as reported in Figure \ref{fig:tradeOffMoralityRateSusceptibleLockdown}, the share of susceptibles after 425 days from the outbreak of the epidemics is substantially higher (82.4\% versus 73.3\%); moreover, as highlighted by Figures \ref{fig:infectionRateItalyBaselineScenariosAddCost0.3exit0.001} and \ref{fig:simMacroSIRModelItalyBaselineScenariosAddCost0.3exit0.001}, it requires a prolonged period of mobility restrictions (almost one year!). In this respect, scenarios with 30\% of additional cost and an exit threshold of 0.5\% or with 50\% of additional costs and an exit threshold of 0.1\% endogenously present a succession of periods with and without mobility restrictions, making this scenario more socially feasible.

\begin{figure}[htbp]
	\centering
	\begin{subfigure}[h]{0.32\textwidth}
		\centering
		\includegraphics[width=1\linewidth]{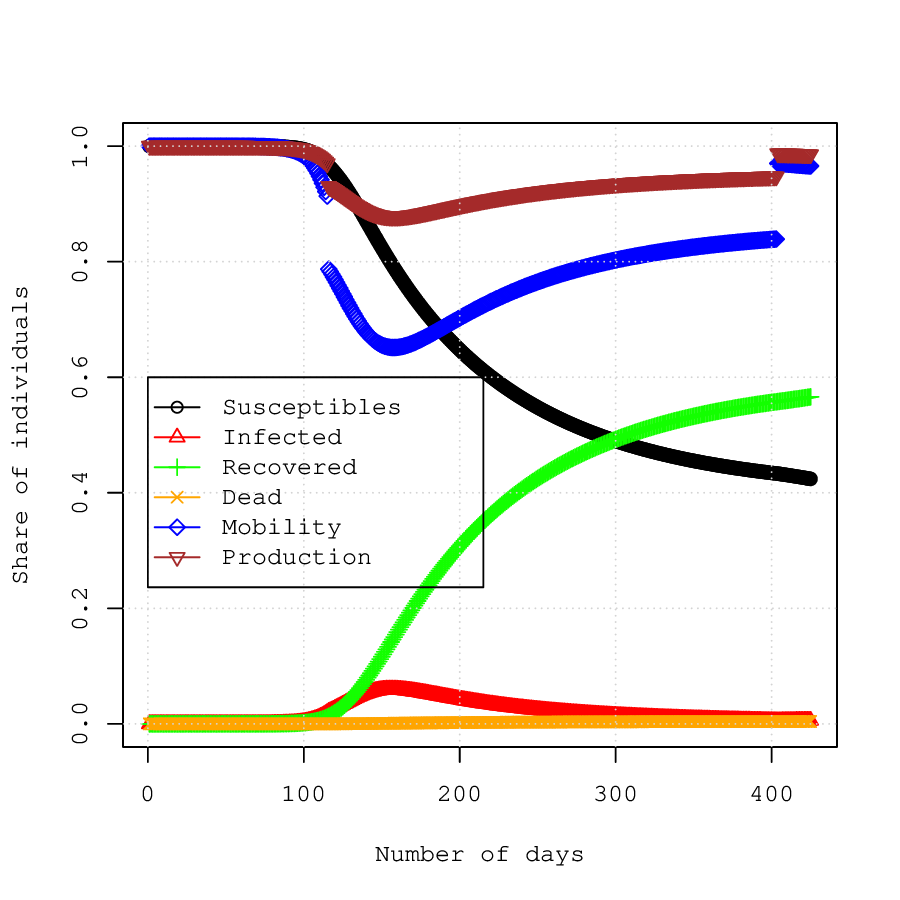}
		\caption{Dynamics when lockdown implies an increase of 10\% of costs of mobility}
		\label{fig:simMacroSIRModelItalyBaselineScenariosAddCost0.1}
	\end{subfigure}
	\begin{subfigure}[h]{0.32\textwidth}
		\centering
		\includegraphics[width=1\linewidth]{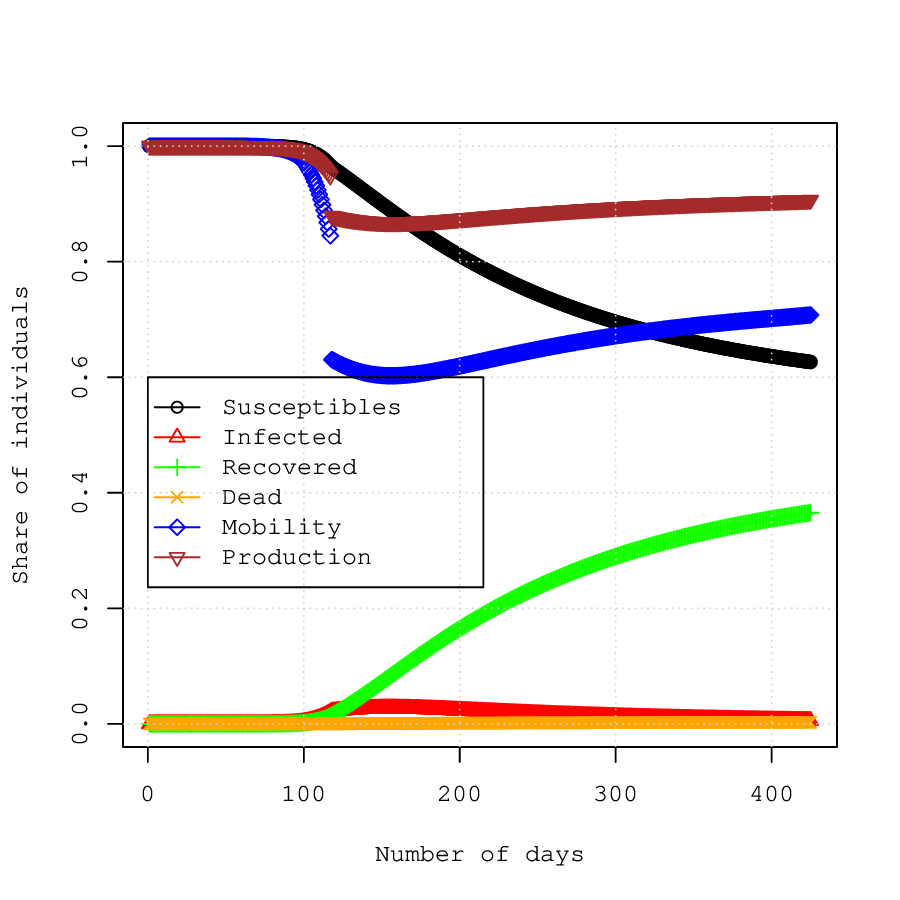}
		\caption{Dynamics when lockdown implies an increase of 20\% of costs of mobility}
		\label{fig:simMacroSIRModelItalyBaselineScenariosAddCost0.20}
	\end{subfigure}
	\begin{subfigure}[h]{0.32\textwidth}
		\centering
		\includegraphics[width=1\linewidth]{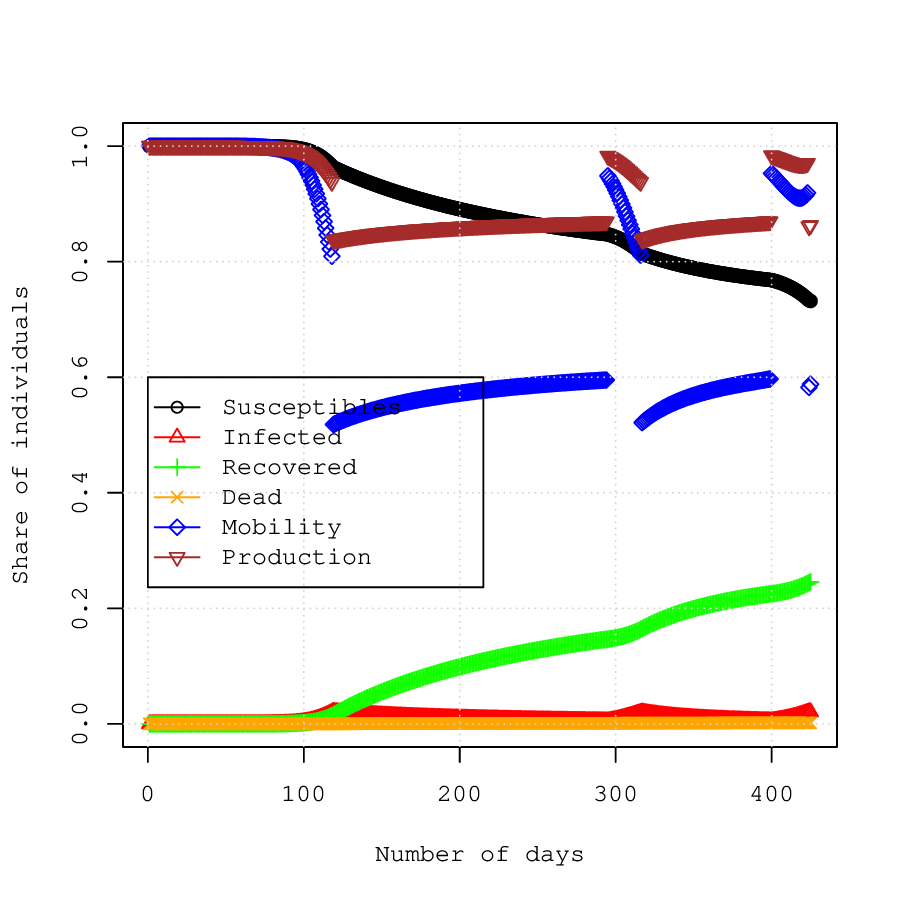}
		\caption{Dynamics when lockdown implies an increase of 30\% of costs of mobility}
		\label{fig:simMacroSIRModelItalyBaselineScenariosAddCost0.3}
	\end{subfigure}
	\begin{subfigure}[h]{0.32\textwidth}
		\centering
		\includegraphics[width=1\linewidth]{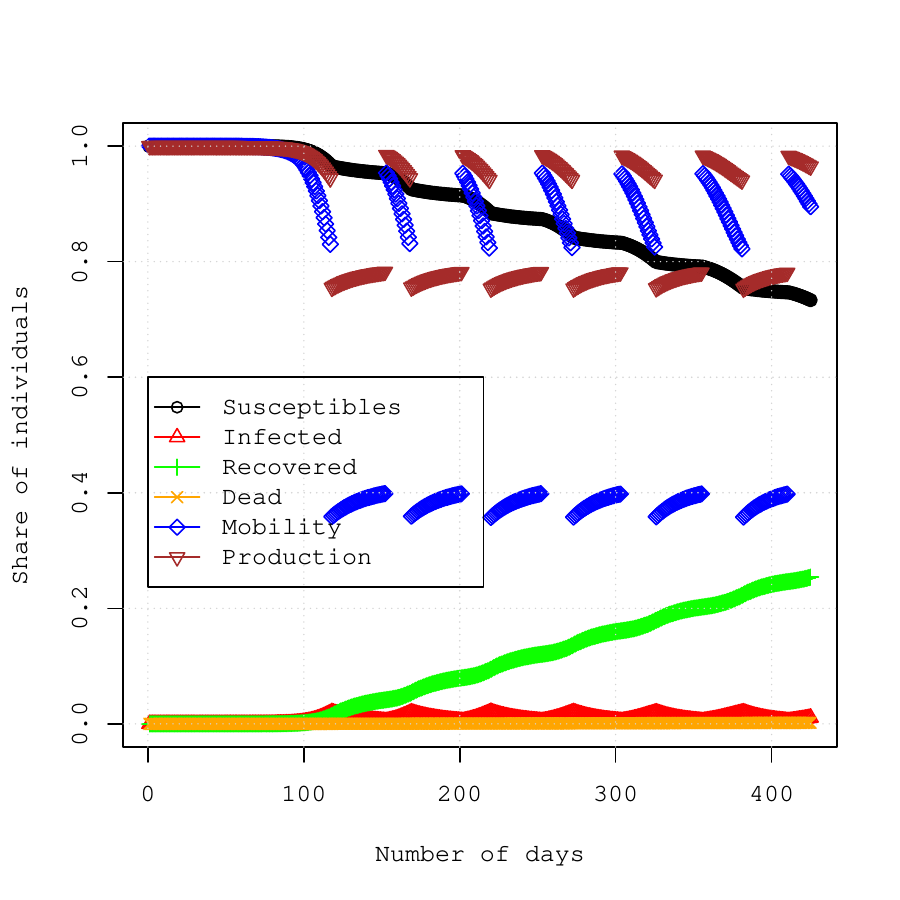}
		\caption{Dynamics when lockdown implies an increase of 50\% of costs of mobility}
		\label{fig:simMacroSIRModelItalyBaselineScenariosAddCost0.5}
	\end{subfigure}
	\begin{subfigure}[h]{0.32\textwidth}
		\centering
		\includegraphics[width=1\linewidth]{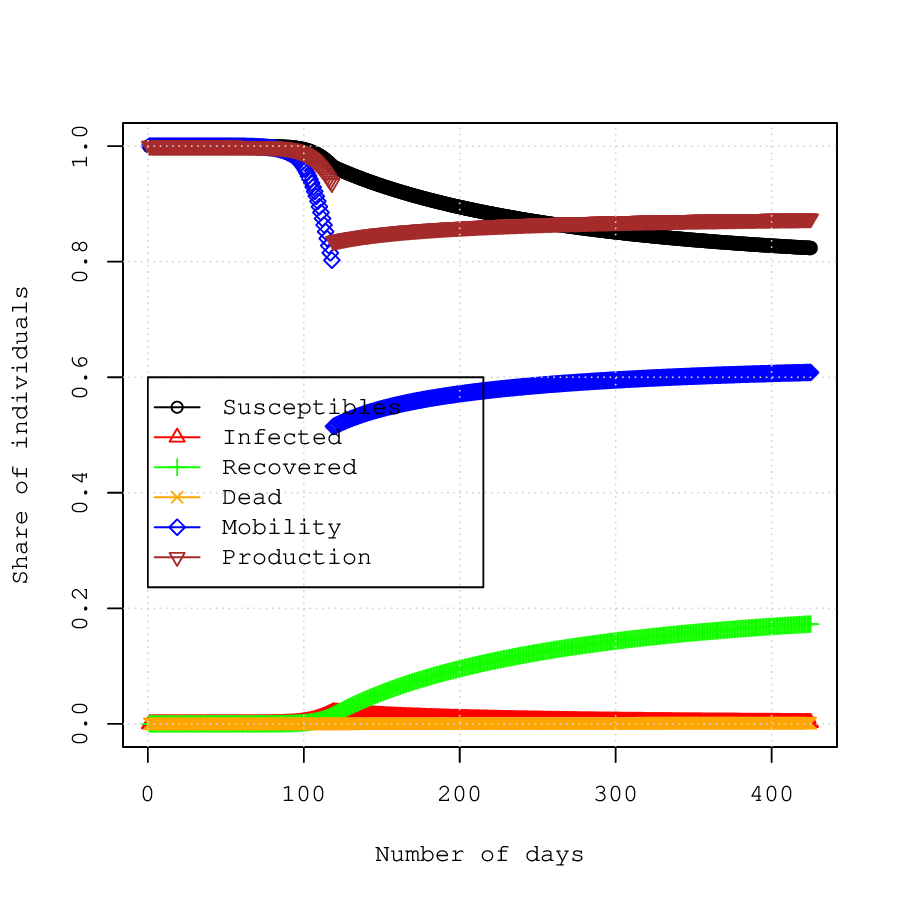}
		\caption{Dynamics when lockdown implies an increase of 30\% of cost of mobility and more restrictive conditions for the exit of lockdown}
		\label{fig:simMacroSIRModelItalyBaselineScenariosAddCost0.3exit0.001}
	\end{subfigure}
	\begin{subfigure}[h]{0.32\textwidth}
		\centering
		\includegraphics[width=1\linewidth]{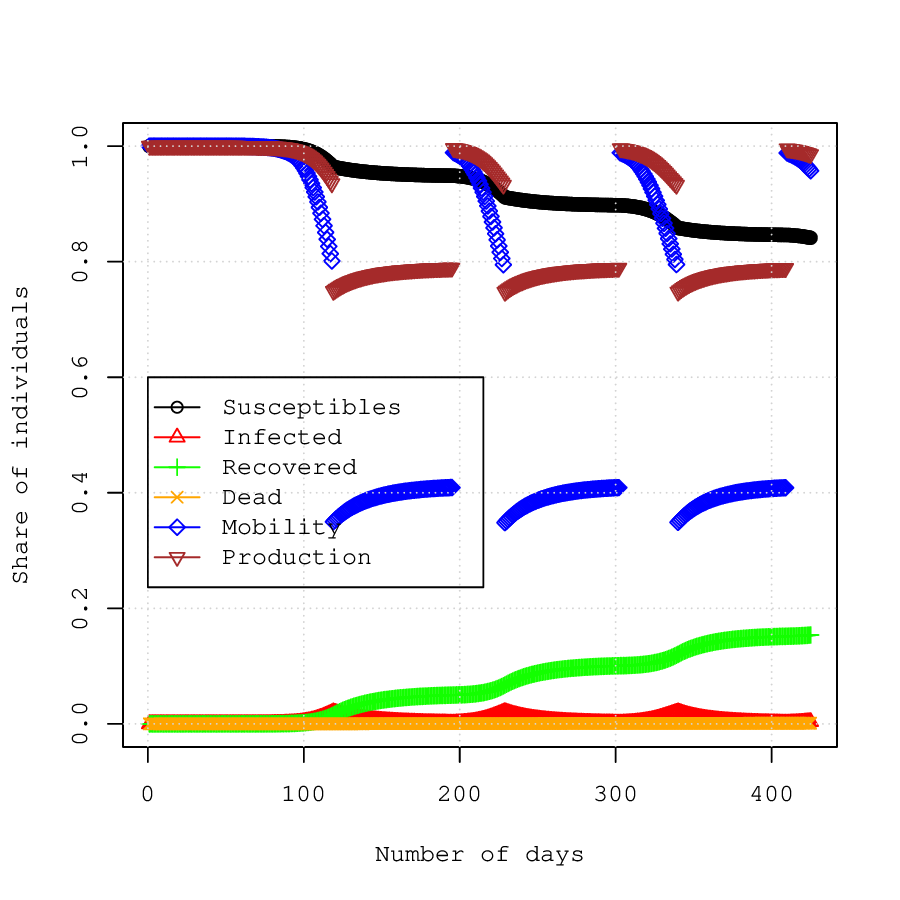}
		\caption{Dynamics when lockdown implies an increase of 50\% of cost of mobility and more restrictive conditions for the exit of lockdown}
		\label{fig:simMacroSIRModelItalyBaselineScenariosAddCost0.5exit0.001}
	\end{subfigure}
	\caption{Dynamics of epidemics and of main economic variables in alternative scenarios of mobility restrictions and exit from these restrictions. Numerical experiments based on the parameters reported in Table \ref{tab:listModelParameters}.}
	\label{fig:dynamicsDifferentScenarios}
\end{figure}

We conclude by observing that, even though individuals are perfectly informed of restriction policy and of the behavior of the pandemic, several scenarios include waves of infections, as a result of the endogenous switching between a regime with mobility restrictions and one without any restrictions (see, e.g., Figures \ref{fig:simMacroSIRModelItalyBaselineScenariosAddCost0.3}-\ref{fig:simMacroSIRModelItalyBaselineScenariosAddCost0.5exit0.001}).

\section{Concluding remarks\label{sec:concludingRemarks}}

We provided a dynamic macroeconomic equilibrium model with pandemic, denoted ESIRD, where perfect-foresight forward looking agents' (short-term) mobility positively affects their income (and consumption), but also contributes to the spread of the pandemic in an extended SIRD model.
Dynamics of economy and pandemic is jointly driven by strategic complementaries in production and negative externalities on infection rates of individual mobilities. We therefore addressed one of the main economic-driven leverages of compartmental epidemiological models, that is, the endogenization of reproduction rate of epidemic \citep{avery2020economist}.
After having proven the existence of a Nash equilibrium and studied the recursive construction of equilibrium(a), we conducted some numerical investigations on the forward-backward system resulting from individual optimizing behavior, calibrating model's parameters on Italian experience on COVID-19 in 2020-2021.

In our ESIRD model, the forward-looking behavior of agents tended to smooth the peak prevalence of pandemic compared to the simplest SIRD model with ``na\"ive'' agents, but in our numerical explorations peak prevalence appeared to be still too high to be sustainable for the Italian health system (e.g., in relation to the number of available beds in hospitals). Once we established that self-regulation of individual mobility decisions was not sufficient to manage the pandemic\footnote{{The model allows to give an answer to the provocative question posed, among others, by \cite{cochrane2020sir} on the viability of a containment policy based only on self-confinement of individuals free of any governmental restrictions on mobility. At least for the Italian experience in 2020, our model suggested that a policy only based on self-confinement would have resulted in a peak prevalence of nearly six million infected people (see Section \ref{sec:questionModel}), which corresponds to a need of about four hundred thousand beds in hospitals. This would have been unsustainable for a country having, in February 2020, about 190,000 beds in hospitals, most of them already occupied by patients with COVID-19 independent pathologies.}}, we evaluated different regimes of mobility restrictions, which could be easily accommodate within our theoretical framework.

In particular, we argued that regimes compatible with the saturation of the healthcare system must be evaluated in terms of a trade-off between economic losses and fatalities as proposes, e.g., by \cite{kaplan2020great, acemoglu2020multi}, but also for their social feasibility of maintaining prolonged periods of mobility restrictions and for leaving higher shares of susceptible at the end of the period, which makes new outbreak of epidemic more likely. In this respect, we pointed out that successive small waves of epidemic can be the result of an efficient regime of mobility restrictions.

Our analysis raises a series of issues for future research.

We ignore heterogeneity of population in terms of ''risk groups'' (typically, in case of COVID-19, age cohorts, see \citealp{salje2020estimating} and \citealp{acemoglu2020multi}), and therefore we cannot evaluate any policy conditioned to individual characteristics, as, for instance, done by \cite{brotherhood2020economic} or \cite{gollier2020cost}.
We also focus on a world before the vaccine, that is standard in this kind of models \citep{boppart2020integrated} and consistent with the period used to calibrate the model. However, in a world with vaccine, or with an expected date of its availability, different questions arise for the timing, targets and costs of vaccination \citep{hung2021single} as well as for the timing of mobility restrictions. Finally, we did not include other non-pharmaceutical interventions, and in particular we do not model testing policies, as, for instance, in \cite{eichenbaum2020macroeconomicstesting}.

Some extensions of empirical analysis appear very promising. Firstly, this is the possibility to study scenarios where mobility restrictions are (mostly) focused on mobility for production or on mobility for consumption. For example, in Europe the second waves of restrictive measures in the period Oct 2020 - May 2021, largely revolved around mobility for consumption.\footnote{See for instance, for France, JORF 0080, 3 April 2021, Text 28, \url{https://www.legifrance.gouv.fr/jorf/id/JORFTEXT000043327303}.}
A second extension concerns the more precise estimation of the relationship between individual mobility, aggregate mobility and production in the presence of strategic complementarities, which poses non trivial issue of identification \citep{manski2000economic}.

We also neglect the possibility of introducing masking and using alternative protective equipment against the epidemic. In case their use is mandatory, it should be equivalent to an exogenous reduction of $\beta_p$ and $\beta_c$ in Eq. (\ref{eq:taudef}) that, by reducing the infection rate, would lead to an increase in the individuals' mobility. Much more complicated is the case in which their use is an individual choice, and their use involves a cost. We should consider a possible free-riding problem because the net benefits of using a mask are decreasing if other individuals are already using a mask.

From the theoretical point of view, we leave open the question of the uniqueness of equilibrium and to obtain stronger properties of the equilibria. A possible answer is to look at the \textit{Master Equation} associated to our model, as suggested in Section 1.4 in \citet{Cardaliaguet2020}.
\begin{small}
	\bibliographystyle{apalike}
	\bibliography{biblioSpatialEpidemic}
\end{small}

\newpage

\appendix

\begin{footnotesize}

\section{Proofs}
\label{app:proofs}

\begin{Theorem}[Tikhonov's fixed point Theorem]\label{th:t}
	Let $\mathcal{V}$ be a locally convex topological vector space, let $\mathcal{Q}\subseteq \mathcal{V}$  be a nonempty compact convex set, and let  $F:\mathcal{Q}\to\mathcal{Q}$ be a continuous function. Then $F$  has a fixed point.
\end{Theorem}
\begin{proof}
	 See, e.g., Theorem (1.10), p.\,147 of \cite{granas2003fixed}.
	\end{proof}

\bigskip

\begin{proof}[Proof of Theorem \ref{th:existence}]
	Fix $\bm\mu(0)\in\mathcal{P}(\K)$ and  $k(0)\in\K$.
Consider the space of sequences
$$
\mathcal{V}:=\bigg\{\bm q=(q_{R},q_{I},q_{S},q_{D}):\N\to \R^{8}  \bigg\}
$$
endowed with the topology of  pointwise convergence. The latter is  a locally convex  topological vector space, since the topology is induced by the family of seminorms
$$\mathbf{p}_{t}(\bm q)=|\bm q(t)|_{{\R^{8}}}, \ \ \ t\in\N,$$
where $\bm q(t)$ is the $t-$th component of $\bm q$.
Then,
consider
$$
\mathcal{Q}:=\bigg\{\bm q=(q_{R},q_{I},q_{S},q_{D}):\N\to [0,1]^{2}\times[0,1]^{2}\times[0,1]^{2}\times\{(0,0)\}  \bigg\}\subset \mathcal{V}.
$$
$\mathcal{Q}$ is convex  and, by Tikhonov's compactness Theorem, it is compact in $\mathcal{V}$. We consider the one-to-one correspondence
$
\mathcal{M}:	\mathcal{Q}\to \mathcal{A}
$
defined by
$$
(\mathcal{M} \bm q)(t,k)\equiv q_{k}(t), \ \ \ (t,k)\in \N\times \K.
$$
Let $\bm\mu^{\bm q}$ be the solution to \eqref{eq:evolu} associated to $\overline{\bm\vartheta}=\mathcal{M}(\bm  q)$ and let
$$F:\mathcal{Q}\to\mathcal{Q}, \ \ \ \ F(\bm q)(t,k):=(\hat{\vartheta}_p(t,k;\bm q),\hat{\vartheta}_c(t,k;\bm q)),  \ \ \ \   (t,k)\in \N\times \K.$$
where  $\hat{\bm\vartheta}(t,k;\bm q)= (\hat{\vartheta}_p(t,k;\bm q),\hat{\vartheta}_c(t,k;\bm q))$ is  the unique  the maximizer over $[0,1]^2$ of
$$ \bm\vartheta\mapsto \sum_{k'\in\K} p_{kk'}(t)\big[ u(t,c(t), k,\bm{\vartheta}(t))
+ (1-\rho)V(t+1,k',\bm\mu^{\bm q}(t+1),(\mathcal{M}\bm q)(t+1,k'))
\big].$$
Clearly, if ${\bm q^{*}}$ is a fixed point of $F$, then $(V(\cdot,\cdot, \bm{\mu}(0),\mathcal{M}({\bm q^{*}})), \mathcal{M}({\bm q^{*}}))$ is an equilibrium according to
Definition \ref{def:eq}.
Given  a sequence $(\bm q^{n})\subset \mathcal{Q}$  converging to $\bm q\in\mathcal{Q}$, we have $$V(t,k, {\bm{\mu}^{\bm q_n}}(t),\mathcal{M}(\bm q_{n})) \to V(t,k, {\bm{\mu}}^{\bm q}(t),\mathcal{M}(\bm q))$$ for each $t\geq 0$. Consequently, by strict concavity and regularity of $\bm\vartheta \mapsto u(t,c(t), k, \bm \vartheta)$,  we  also have the convergence
$\bm{\hat{\vartheta}}(t,k;\bm q_{n})\to \bm{\hat{\vartheta}}(t,k;\bm q).$
This shows that $F$ is continuous. We conclude by Theorem \ref{th:t}.
\end{proof}

\bigskip

\begin{proof}[Proof of Theorem \ref{thm:verif}]
	We show that (i) and (ii) of Definition \ref{def:eq} hold for the couple  $(\hat{v},\hat{\bm\vartheta})$, which is assumed to be well defined by induction (as $\hat\xi$ is so at each step).

	We preliminarily notice that, given $(t_{0},k(t_0))\in\N\times \{S,I,R\}$, the function $$[0,1]^2\to \R, \ \ \bm\vartheta= (\vartheta_p, \vartheta_c)  \mapsto u(t_{0},c(t_{0}),k(t_0),\bm\vartheta)$$ is strictly concave, since
	$$
	D_{\bm\vartheta}
	u(t_{0},c(t_{0}),k(t_0),\bm\vartheta)=
	\left(
	\frac{A_1^{k(t_0)} }{A_0^{k(t_0)} + A_1^{k(t_0)} \vartheta_p}
	- \gamma_p(t_{0}, k(t_0),\mu(t_{0})), \
	\frac{P_{1} }{P_{0} + P_{1} \vartheta_c}
	- \gamma_c(t_{0, }k(t_0),\mu(t_{0}))\right),
	$$
	and
	$$
	D^2_{\bm\vartheta}
	u(t_{0},c(t_{0}),k(t_0),\bm\vartheta)=
	\left(
	\begin{array}{cc}\displaystyle
	-\frac{(A_1^{k(t_0)})^2}{(A_0^{k(t_0)} + A_1^{k(t_0)} \vartheta_p)^2}
	& 0 \\
	0 & \displaystyle -\frac{P^2_{1} }{(P_{0} + P_{1} \vartheta_c)^2}
	\end{array}
	\right).
	$$
	Now we fix $t_0\in\N$ and show that $\hat{v}(t_0,\cdot)$ solves the dynamic programming equation on the various occurrences of $k(t_0)\in\K$ and that  $\hat{\bm\vartheta}(t_0,\cdot)$ defined as in the algorithm are  the maximizers of the right hand side of \eqref{DPE} .
	\begin{itemize}
		\item \textbf{Case $\bm{k(t_0)=D}$}. In this case the unique admissible control is ${\bm\vartheta}(t_{0},D):=(0,0)$ and the Bellman  equation reduces to
		\begin{align}\label{eq:DPER}
		&v(t_{0},D) =
		u(t_{0},0,D,(0,0))
		+ (1-\rho)v(t_{0}+1,D) = (1-\rho)v(t_{0}+1,D).
		\end{align}
		It is clear that the above constructed $\hat v$ is always zero on $D$ and hence satisfies the above equation. The maximizer ${\hat{\bm\vartheta}}(t_{0},D)$ is the unique admissible control, i.e.   $\hat{\bm\vartheta}(t_{0},D)=(0,0)$.
		
		\medskip
		
		\item \textbf{Case $\bm{k(t_0)=R}$}. In this case the Bellman equation reduces to
		\begin{align}\label{eq:DPERbis}
		v(t_{0},R)& = \sup_{\bm\vartheta\in [0,1]^2}
		\big( u(t_{0},c(t), R,\bm\vartheta)
		+ (1-\rho)v(t_{0}+1,R) \big)=  (1-\rho) v(t_{0}+1,R) +\sup_{\bm\vartheta\in [0,1]^2}
		u(t_{0},c(t), R,\bm\vartheta).
		\end{align}
		The optimization above leads to the unique maximum point
		$$\bm{\hat \vartheta}=(\hat{\vartheta}_{p},\hat{\vartheta}_{c})= \big((\tilde{\vartheta}_{p}\wedge 1)\vee 0,\ (\tilde{\vartheta}_{c}\wedge 1)\vee 0	\big),$$
		where
		$$
		\begin{cases}
		\displaystyle{\tilde{\vartheta}_p=
			\frac{A_{1}^R-\gamma_p(t_{0},I,\bm\mu(t_{0}))A_{0}^R}
			{\gamma_p(t_{0}, R,\bm\mu(t_{0}))A_{1}^R}
			=\frac{1}{\gamma_p(t_{0},R,\bm\mu(t_{0}))}-\frac{A_{0}^R}{A_{1}^R}},\\\\
		\displaystyle{\tilde{\vartheta}_c=
			\frac{P_{1}-\gamma_c(t_{0},R,\bm\mu(t_{0}))P_{0}}
			{\gamma_c(t_{0},R,\bm\mu(t_{0}))P_{1}}
			=\frac{1}{\gamma_c(t_{0},R,\bm\mu(t_{0}))}-\frac{P_{0}}{P_{1}}.}
		\end{cases}
		$$
		We therefore get
		$$v(t_{0}+1,R) = \frac{v(t_{0},R)- u(t_{0},c(t_{0}),R,\bm{\hat\vartheta})}{1-\rho}.$$
		Hence $\hat v(t_{0},\cdot)$ defined as in \eqref{def:hatv} satisfies by construction the Bellman equation \eqref{DPE}  with
	 maximizer  $\hat{\bm\vartheta}(t_{0},R)$ given by \eqref{thetaRI}.
		
		\medskip
		
		\item \textbf{Case $\bm{k(t_0)=I}$}. In this case the dynamic programming equation reduces to
		\begin{align}\label{eq:DPEI}
		v(t_{0},I) &= \sup_{\vartheta\in [0,1]^2}
		\Big( u(t_{0},c(t_{0}),I,\bm\vartheta)
		+ (1-\rho)\left((1-\pi_R-\pi_D)v(t_{0}+1,I) +\pi_Rv(t_{0}+1,R)\right)\big)\nonumber\\
		& =  (1-\rho)\left((1-\pi_R-\pi_D)v(t_{0}+1,I) +\pi_Rv(t_{0}+1,R)\right) + \sup_{\vartheta\in [0,1]^2}
		u(t_{0},c(t_{0}),I,\bm\vartheta).
		\end{align}
		The optimization above leads to the unique maximum point
		$$(\hat{\vartheta}_{p},\hat{\vartheta}_{c})= \big((\tilde{\vartheta}_{p}\wedge 1)\vee 0,\ (\tilde{\vartheta}_{c}\wedge 1)\vee 0	\big),$$
		where
		$$
		\begin{cases}
		\displaystyle{\tilde{\vartheta}_p=
			\frac{A_{1}^I-\gamma_p(t_{0},I,\hat{\bm\mu}(t_{0}))A_{0}^I}
			{\gamma_p(t_{0},I,\hat{\bm\mu}(t_{0}))A_{1}^I}
			=\frac{1}{\gamma_p(t_{0},I,\hat{\bm\mu}(t_{0}))}-\frac{A_{0}^I}{A_{1}^I}},\\\\
		\displaystyle{\tilde{\vartheta}_c=
			\frac{P_{1}-\gamma_c(t_{0},I,\hat{\bm\mu}(t_{0}))P_{0}}
			{\gamma_c(t_{0},I,\hat{\bm\mu}(t_{0}))P_{1}}
			=\frac{1}{\gamma_c(t_{0},I,\hat{\bm\mu}(t_{0}))}-\frac{P_{0}}{P_{1}}.}
		\end{cases}
		$$
		We therefore get
		$$v(t_{0}+1,I)=\frac{1}{1-\pi_R-\pi_D}\left[\frac{v(t_{0},I)-
			u(t_{0},c(t_{0}),I,\bm{\hat\vartheta})}{1-\rho}- \pi_Rv(t_{0}+1,R) \right].$$
Hence $\hat v(t_{0},\cdot)$ defined as in \eqref{def:hatv} satisfies by construction the Bellman equation \eqref{DPE} with 	 maximizer  $\bm{\hat\vartheta}(t_{0},I)$ given by \eqref{thetaRI}. 		
		\medskip
		
		\item \textbf{Case $\bm{k(t_0)=S}$}. In this case the Bellman equation reduces to
		\begin{align}\label{eq:DPES}
		v(t_{0},S) &= \sup_{\bm\vartheta\in [0,1]^2}
		\Big( u(t_{0},c(t_{0}),S,\bm\vartheta)
		+ (1-\rho)\left( (1- \tau(t_{0})) v(t_{0}+1,S) +  \tau(t_{0}) v(t_{0}+1,I)\right)\big),
		\end{align}
		which can be rewritten as
		\begin{align}\label{eq:DPESbis}
		v(t_{0},S) =& \ \ (1-\rho) v(t_{0}+1,I)+(1-\rho) (v(t_{0}+1,S)-v(t_{0}+1,I)) \\&+
		\sup_{\bm\vartheta\in [0,1]^2}
		\Big( u(t_{0},c(t_{0}),S,\bm\vartheta)
		- (1-\rho)  \tau(t_{0})  (v(t_{0}+1,S)-v(t_{0}+1,I))\Big),
		\end{align}
	\end{itemize}
	Set $\xi:= v(t_{0}+1,S)-v(t_{0}+1,I)$ and consider the optimization above in terms of the parameter $\xi\in\R_{+}$.
	The maximization
	leads to the unique maximum point
	$$\bm{\hat{\vartheta}}^{\xi}=(\hat{\vartheta}_{p}^{\xi},\hat{\vartheta}_{c}^\xi)= \big((\tilde{\vartheta}_{p}^\xi\wedge 1)\vee 0,\ (\tilde{\vartheta}_{c}^\xi\wedge 1)\vee 0	\big),$$
	where
	\begin{eqnarray*}
		\tilde\vartheta_p^\xi=
		\frac{1}{\gamma_p(t_{0},S,\hat{\bm\mu}(t_{0}))+{(1-\rho)}  \hat a(t_{0})\xi}
		-\frac{A_{0}^S}{A_{1}^S}, \ \ \ \ \ \ \ \
		\tilde\vartheta_c^\xi=
		\frac{1}{\gamma_c(t_{0},S,\hat{\bm\mu}(t_{0}))+{(1-\rho)}  \hat{b}(t_{0})\xi}
		-\frac{P_{0}}{P_{1}},
	\end{eqnarray*}
	where
	\begin{equation}\nonumber
	\hat a(t_{0})= \hat\mu(t_{0},I)\hat\vartheta_{p}(t_{0},I), \ \ \
	\hat b(t_{0})= \hat\mu(t_{0},I)\hat\vartheta_{c}(t_{0},I).
	\end{equation}
	Recalling the definition of $f$ given in \eqref{def:f},
	the Bellman equation reduces to  the algebraic equation in the variable $\xi\in\R^{+}$
	$$
	v(t_{0},S) =(1-\rho) v(t_{0}+1,I)+(1-\rho) \xi + f(t,\xi).
	$$
	By assumption this equation has a unique solution $\hat \xi$.
 Hence $\hat v(t_0,\cdot)$ defined as in \eqref{def:hatv} satisfies by construction the Bellman equation \eqref{DPE} 	with  maximizer  $\bm{\hat\vartheta}(t_{0},S)$
given by \eqref{thetaS}. \end{proof}

\newpage

\section{Procedure of simulation \label{app:simulationProcedure}}

NOTE: In this appendix the notation is lightened from that used in the body of the article to avoid making the formulas too heavy thinking and difficult to read.

\smallskip

\begin{enumerate}[(A)]

\item \textbf{The maximum in the lifetime utilities.}
As $t$ goes to infinitum the number of infected agents converges to zero, i.e. $\lim_{t \to \infty}\mu\left(t,I\right)=0$ and $\lim_{t \to \infty}\mu\left(t,R\right)\gg 0$ and the lifetime utilities is maximum in this state of no pandemic. Then:
\begin{align} \nonumber
 U(S)^{\max} = \lim_{t \to \infty} U(t,S) & =  \lim_{t \to \infty} \frac{\kappa(t,\mu(t,I),SR) + \ln(1+{Z}(t)) }{{\rho}} = \\ \nonumber
=& \dfrac{\kappa(\infty,0,SR) + \ln\left(1+ 1/\gamma_p(\infty,0,SR) - A^{SR}_0/A^{SR}_1 \right)}{{\rho}}; \\ \nonumber
 U(R)^{\max} = \lim_{t \to \infty}  U(t,R) & = \lim_{t \to \infty} \frac{\kappa(t,\mu(t,I),SR) + \ln(1+{Z}(t)) }{{\rho}} = \\ \nonumber
=& \dfrac{\kappa(\infty,0,SR) + \ln\left(1+ 1/\gamma_p(\infty,0,SR) - A^{SR}_0/A^{SR}_1 \right)}{{\rho}}; \\ \nonumber
U(I)^{\max} = \lim_{t \to \infty} U(t,I) & = \dfrac{{\rho}\kappa(\infty,0,I) + {(1-\rho)} \pi_R \kappa(\infty,0,SR)}{{\rho}\left[1- {(1-\rho)}\left(1-\pi_R-\pi_D\right)\right]} + \\ \nonumber
+ & \dfrac{\left[1-{(1-\rho)}\left(1-\pi_R\right)\right]\ln\left(1+ 1/\gamma_p(\infty,0,SR) - A^{SR}_0/A^{SR}_1 \right)}{{\rho}\left[1- {(1-\rho)}\left(1-\pi_R-\pi_D\right)\right]}.
\end{align}
where:
\begin{equation} \nonumber
\kappa(t,\mu_I,SR) := \ln\left(\dfrac{A^{SR}_1}{\gamma_p(t,\mu_I,SR)}\right) + \gamma_p(t,\mu_I,SR)\dfrac{A^{SR}_0}{A^{SR}_1} + \ln\left(\dfrac{P_1}{\gamma_c(t,\mu_I,SR)}\right) +\gamma_c(t,\mu_I,SR)\dfrac{P_0}{P_1} - 2;
\end{equation}
and
\begin{equation} \nonumber
\kappa(t,\mu_I,I) := \ln\left(\dfrac{A^{I}_1}{\gamma_p(t,\mu_I,I)}\right) + \gamma_p(t,\mu_I,I)\dfrac{A^{I}_0}{A^{I}_1} + \ln\left(\dfrac{P_1}{\gamma_c(t,\mu_I,I)}\right) +\gamma_c(t,\mu_I,I)\dfrac{P_0}{P_1} - 2.
\end{equation}

\item \textbf{The feasible set of individual lifetime utilities.}
From Point (A), together with the appropriate choice of $M$ in order to make $U(t, k) \geq 0 \; \text{for } t\geq 0 \text{ and } \forall k \in \K$, the feasible set of individual lifetime utilities is defined as follows:
\begin{align}
T:= \left\{ \left(x,y,z\right) \in \left(0,U(R)^{\max}\right) \times \left(0,U(I)^{\max}\right) \times \left(0,U(R)^{\max}\right): y \leq x \leq z \right\}.
\label{eq:feasibleRangeEquilibriumTrajectories}
\end{align}
This gives a bound for the lifetime utilities in the spirit of Theorem \ref{thm:verif}.
\item \textbf{Set the health status distribution of population at time $0$ as:}
\begin{align} \nonumber
\mu(0,S) & = 1-\epsilon; \\ \nonumber
\mu(0,I) & = \epsilon; \\ \nonumber
\mu(0,R) & = 0; \\ \nonumber
\mu(0,D) & = 0,
\end{align}
with $\epsilon$ very small.

\item \textbf{Set the initial value of utilities in the three states in the feasible set $T$ by choosing $\delta^I, \delta^S, \delta^R \geq 0$.}

\begin{align} \nonumber
U(0,R) & = U(R)^{\max}(1-\delta^R); \\ \nonumber
U(0,S) & = U(0,R)(1-\delta^S); \\ \nonumber
U(0,I) & = U(0,S)(1-\delta^I)
\end{align}

\item \textbf{Calculate $a(0)$ and $b(0)$}:
\begin{align} \nonumber
a(0) & =\beta_p \times \mu(0,I) \times \overline\vartheta_p(0,\mu(0,I),I) \\ \nonumber
b(0) & =\beta_c \times \mu(0,I) \times \overline\vartheta_c(0,\mu(0,I),I),
\end{align}
where
\begin{align} \nonumber
\overline\vartheta_p(0,\mu(0,I),I) & = \dfrac{1}{\gamma_{p}(0,\mu(0,I),I)} - \dfrac{A^I_0}{A^I_1} \text{ and}\\ \nonumber
\overline\vartheta_c(0,\mu(0,I),I) & = \dfrac{1}{\gamma_{c}(0,\mu(0,I),I)} - \dfrac{P_0}{P_1}.
\end{align}

\item \textbf{Find $\Delta U(1,S,I) := U(1,S) - U(1,I)$ by solving the following implicit equation}
\begin{align} \nonumber
0 &= -{(1-\rho)} \left(1-\pi_R -\pi_D\right) \Delta U(1,S,I) + \left(1-\pi_R -\pi_D\right) U(0,S) - U(0,I) +\pi_R U(0,R) + \\ \nonumber
& -\pi_R \kappa(1,\mu(0,I),R)  + \kappa(1,\mu(0,I),I) - \left(1-\pi_R -\pi_D\right)\chi\left(\Delta U(1,S,I)\right) + \pi_D \ln \left(1 + Z\left(\Delta U(1,S,I) \right)\right),
\end{align}
where
\begin{align} \nonumber
\chi\left( \Delta U(1,S,I)\right) &:= \\ \nonumber
& \ln\left(\dfrac{A^{SR}_1}{\gamma_p(1,\mu(0,I),S) + {(1-\rho)} a(0) \Delta U(1,S,I)}\right) + \\ \nonumber &+\dfrac{A^{SR}_0}{A^{SR}_1}\left\{\gamma_p(1,\mu(0,I),S) + {(1-\rho)} a(0)  \Delta U(1,S,I)\right\} + \\ \nonumber
&+ \ln\left(\dfrac{P_1}{\gamma_c(1,\mu(0,I),S) + {(1-\rho)} b(0) \Delta U(1,S,I)}\right) + \\ \nonumber
&+ \dfrac{P_0}{P_1}\left\{\gamma_c(1,\mu(0,I),S) + {(1-\rho)} b(0)\Delta U(1,S,I)\right\} -2 \\ \nonumber
\end{align}
and
\begin{align} \nonumber
{Z}(0) = Z\left(\Delta U(1,S,I) \right) & =  \mu(0,S) \times \left[\frac{1}{\gamma_p(1,\mu(0,I),S)+{(1-\rho)} a(0) \Delta U(1,S,I)}-\frac{a^{SR}_{0}}{a^{SR}_{1}}\right] + \\ \nonumber
& + \mu(0,I) \times \overline\theta_p(0,\mu(0,I),I) + \mu(0,R) \times \overline\theta_p(0,\mu(0,I),R).
\end{align}
where
\begin{align} \nonumber
\overline\theta_p(0,\mu(0,I),R) & = \dfrac{1}{\gamma_{p}(0,\mu(0,I),R)} - \dfrac{a^{SR}_0}{a^{SR}_1} \text{ and}\\ \nonumber
\overline\theta_c(0,\mu(0,I),R) & = \dfrac{1}{\gamma_{c}(0,\mu(0,I),R)} - \dfrac{P_0}{P_1},
\end{align}
where we set $\mu(1,k) \approx \mu(0,k) \; \forall k \in \K$, to simplify the calculations. This approximation is more and more accurate as time scale of simulation is smaller, in the limit of continuos time is exact.

\item \textbf{Calculate the movement of susceptible}
\begin{align} \nonumber
\overline\vartheta_p(0,\mu(0,I),S) &=\frac{1}{\gamma_p(0,\mu(0,I),S)+{(1-\rho)} a(0) \Delta U(1,S,I)}-\frac{A^{SR}_{0}}{A^{SR}_{1}}; \\ \nonumber
\overline\vartheta_c(0,\mu(0,I),S) &=\frac{1}{\gamma_c(0,\mu(0,I),S)+{(1-\rho)} b(0) \Delta U(1,S,I)} -\frac{P_{0}}{P_{1}}.
\end{align}

\item \textbf{Calculate the level of lifetime utilities at time 1}
 \begin{align} \nonumber
 U(1,R) & = \dfrac{U(0,R) -\ln \left(1 + {Z}\left(0\right)\right) - \kappa(0,\mu(0,I),R)}{{1-\rho}}; \\ \nonumber
 U(1,I) & = \dfrac{U(0,I)-\pi_R U(0,R) +\pi_R \kappa(0,\mu(0,I),R) -\kappa(0,\mu(0,I),I) - \left(1-\pi_R\right)\ln \left(1 + {Z}\left(0\right)\right)}{{(1-\rho)}\left(1-\pi_R-\pi_D\right)} ;\\ \nonumber
 U(1,S) & = \Delta U(1,S,I) + U(1,I) ;
 \end{align}
 \item \textbf{Upgrade the health status distribution of population at time 1}
 \begin{align} \nonumber
 \mu(1,S) &= \mu(0,S) \left[1-a(0)\overline{\vartheta}_p(0,\mu(0,I),S)- b(0)\overline{\vartheta}_c(0,\mu(0,I),S)\right],\\ \nonumber
 \mu(1,I) &= \mu(0,S)\left[a(0)\overline{\vartheta}_p(0,\mu(0,I),S)+ b(0)\overline{\vartheta}_c(0,\mu(0,I),S)\right]+\mu(0,I)(1-\pi_R-\pi_D),\\ \nonumber
 \mu(1,R) &= \mu(0,I)\pi_R+\mu(0,R),\\ \nonumber
 \mu(1,D) &= \mu(0,D) + \mu(0,I)\pi_D.
 \end{align}

 \item \textbf{Check if Condition (\ref{eq:feasibleRangeEquilibriumTrajectories}) is satisfied. If not start with a new set of $\delta$s at point D. If Condition (\ref{eq:feasibleRangeEquilibriumTrajectories}) is satisfied and the number of periods is lower of a given threshold repeat points E-I by taking the new level of $\mu$s at point I.}

\end{enumerate}

\end{footnotesize}

\end{document}